\documentclass[twocolumn,letter]{IEEEtran}
\usepackage{graphicx,amssymb,fancybox}
\usepackage{latexsym}
\usepackage{amsmath}

\def\rE{{\rm E}}

\def\FF{\mathbb F}
\def\RR{\mathbb R}

\def\NN{\mathbb N}

\newcommand{\Tr}{{\rm Tr}\,}
\def\mix{\mathop{\rm mix}}
\def\Ker{\mathop{\rm Ker}}
\def\Pr{{\rm Pr}}

\newtheorem{Lmm}{Lemma}
\newtheorem{Thm}{Theorem}
\newtheorem{Dfn}{Definition}
\newtheorem{Crl}{Corollary}
\newtheorem{rem}{Remark}
\newtheorem{proposition}{Proposition}

\def\Label#1{\label{#1}\ \text{[\ #1\ ]}\ }
\def\Label{\label}

\begin{document}
\title{Dual universality of hash functions and its applications to quantum cryptography}
\author{Toyohiro Tsurumaru and Masahito Hayashi
\thanks{
The material in this paper was presented in part
at QCRYPT 2011: First Annual Conference on Quantum Cryptography, Zurich, Switzerland, September 2011.
T. Tsurumaru is with
Mitsubishi Electric Corporation, Information Technology R\&D Center,
5-1-1 Ofuna, Kamakura-shi, Kanagawa, 247-8501, Japan (e-mail: Tsurumaru.Toyohiro@da.MitsubishiElectric.co.jp).
M. Hayashi is with Graduate School of Information Sciences, Tohoku University, Aoba-ku, Sendai, 980-8579, Japan, and
Centre for Quantum Technologies, National University of Singapore, 3 Science Drive 2, Singapore 117542.
(e-mail: hayashi@math.is.tohoku.ac.jp)}}

\maketitle

\begin{abstract}
In this paper, we introduce the concept of dual universality of hash functions and present its applications to quantum cryptography.
We begin by establishing the one-to-one correspondence between a linear function family ${\cal F}$ and a code family ${\cal C}$, and thereby defining $\varepsilon$-almost dual universal$_2$ hash functions, as a generalization of the conventional universal$_2$ hash functions.
Then we show that this generalized (and thus broader) class of hash functions is in fact sufficient for the security of quantum cryptography.
This result can be explained in two different formalisms.
First, by noting its relation to the $\delta$-biased family introduced by Dodis and Smith, we demonstrate that Renner's two-universal hashing lemma is generalized to our class of hash functions.
Next, we prove that the proof technique by Shor and Preskill can be applied to quantum key distribution (QKD) systems that use our generalized class of hash functions for privacy amplification.
While Shor-Preskill formalism requires an implementer of a QKD system to explicitly construct a linear code of the Calderbank-Shor-Steane type,
this result removes the existing difficulty of the construction a linear code of CSS code 
by replacing it by 
the combination of an ordinary classical error correcting code and our proposed hash function.
We also show that a similar result applies to the quantum wire-tap channel.

Finally we compare our results in the two formalisms and show that, in typical QKD scenarios,
the Shor-Preskill--type argument gives better security bounds
in terms of the trace distance and Holevo information, 
than the method based on the $\delta$-biased family.
\end{abstract}

\section{Introduction}

Extracting secure uniform random number is
an important task for cryptographic applications
with the presence of quantum leaked information as well as 
that of classical leaked information.
For the quantum setting, several extractors
are proposed, e.g., 
2-universal hashing \cite{Renner}, 
approximate 2-universal hashing \cite{TSSR11},
sample-and-hash \cite{KR07}, 
one-bit extractors \cite{KT08}, and 
Trevisan's extractor \cite{DPVR09}.
In this paper, we 
focus on universal$_2$ hash functions \cite{CW79}
which has a variety of 
cryptographic applications, for example, for the information theoretically secure signatures, 
the hash functions for 
for privacy amplification \cite{S02,BBCM,ILL}
and for the wire-tap channel\cite{Hayashi2,Hayashi5}.
The class of universal$_2$ hash function families 
is the largest class of families of hash functions 
among known classes of families of hash functions guaranteeing the strong security.
However, there might exist a larger class of 
hash functions guaranteeing the strong security.
If such a class exists, we might realize a strongly secure privacy 
amplification with a smaller complexity.
It is known that 
the class of universal$_2$ hash functions 
is included in the class of 
$\varepsilon$-almost universal$_2$ hash functions\cite{CW79,WC81}.
However, as is shown in Section \ref{s8-b}, 
there exists an example of
$\varepsilon$-almost universal$_2$ hash functions that cannot yield
the strong security.
Hence, we have to consider another type of generalization of
the class of universal$_2$ hash functions. 

In this paper, 
in order to seek such a larger class, we restrict our hash functions to linear functions on a finite-dimensional space over the finite field $\FF_2$
because a larger part of hash functions with a smaller complexity are linear.
Under the restriction, we can find a one-to-one correspondence between 
a hash function and a linear code by considering the kernel of the hash function.
Focusing on the dual code of the code corresponding to the given hash function,
we propose the class of 
$\varepsilon$-almost dual universal$_2$ hash functions
as a 
class of families of linear hash functions
satisfying the following conditions:
\begin{enumerate}
\item The class of families of hash functions contains the class of
universal$_2$ hash functions.
\item 
Any family of hash functions in this class yields the strong security
when the generating key rate is sufficiently small.
\end{enumerate}
%We also show that this new class contains the class of linear hash functions.
%Indeed, this new class of families of hash functions can yield the strong security,
%i.e., when we apply $\varepsilon$-almost dual universal$_2$ hash functions in the privacy amplification with a sufficient sacrifice bits, the resulting bits are strongly secure.

Hence, the relation among class of families of hash functions
is summarized as Fig. \ref{fig:inclusion-universality}.

\begin{figure}[htbp]
\begin{center}
\scalebox{1.0}{\includegraphics[scale=0.4]{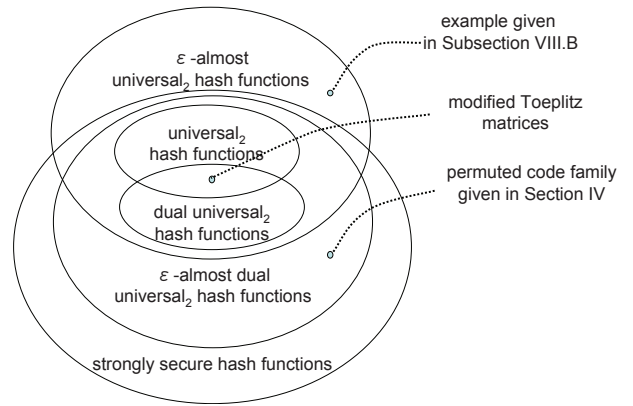}}
\end{center}
\caption{Relation among hash functions (when $\varepsilon$ increases as a polynomial of $n$).
The modified Toeplitz matrices are given by a concatenation $(X, I)$ of the Toeplitz matrix $X$ and the identity matrix $I$, mentioned in Section \ref{sec1}.
}
\Label{fig:inclusion-universality}
\end{figure}%

This fact can be shown by two different approaches.
In the first approach, 
we focus on the concept of the $\delta$-biased family,
which was introduced by Dodis and Smith \cite{DS05}.
Their results have also been extended to the quantum case (\cite{FS08}, or Lemma \ref{Lem6-1-q} of this paper).
Since the main purpose of their original results is to correct errors without leaking partial information,
they do not treat hash functions and privacy amplification.
In this paper, adding an appropriate discussion to their results concerning 
the $\delta$-biased family,
we show the strong security for the case where 
$\varepsilon$-almost dual universal$_2$ hash functions are applied 
in the privacy amplification with a sufficient sacrifice bits.
Since the bound (Lemma \ref{Lem6-3-q}) derived by this approach has a form similar to that by
Renner \cite{Renner},
we need to apply the method of smoothing \cite{cq-security}.
We call this approach the $\delta$-biased approach.

In the second approach,
we focus on the relation between the phase error probability and the leaked information given by the security proof \cite{H07,Hayashi3,Renes10} 
of a QKD protocol called the Bennett-Brassard 1984 (BB84) protocol \cite{BB84}.
The key point of this approach is the error correction in the phase basis by using a certain type of random coding.
Hence we call this approach the phase error correction approach.

While both approaches derive similar conclusions qualitatively,
the security bounds are different even when the same $\varepsilon$-almost dual universal$_2$ hash functions are applied.
In this respect, the phase error correction approach has two advantages over 
the $\delta$-biased approach.
As the first advantage, 
in the case of the BB84 QKD protocol via a depolarizing channel,
as is shown in Section \ref{s6-b},
the phase error correction approach yields better bounds
in terms of the trace distance and Holevo information, 
than the $\delta$-biased approach.

Next in order to explain the second advantage, we consider the case
where we apply the privacy amplification after the error correction.
In this setting, we treat a pair of two codes, i.e., 
the larger code for the error correction, and the smaller code for privacy amplification.
Then the second advantage of the phase error correction approach is that it can guarantee the strong security with a larger class of families of code pairs, than the $\delta$-biased approach.
In fact,  in order to guarantee the strong security in this setting, the $\delta$-biased approach requires 
$\varepsilon$-almost dual universal$_2$ hash functions
for a fixed error correction code.
However, 
in the phase error correction approach, we can relax this requirement for the family of code pairs.
That is, this approach guarantees the strong security with a larger class 
of families of code pairs.
As a concrete example of advantage of this concept, 
we note the construction of an appropriate deterministic hash function for a given error correction code, 
which needs the treatment of the security for such a larger class 
of families of code pairs.
That is, employing the phase error correction approach, 
we can show the existence of a deterministic hash function for 
a given error correction code that is
universally secure under the independent and identical condition.

The organization of this paper is as follows.
We begin in Section \ref{sec1} by reviewing the conventional universal hash functions, i.e., the properties of $\varepsilon$-almost universal$_2$ functions. Then we restrict ourselves to linear hash functions over a finite field $\FF_2^n$, and establish a one-to-one correspondence between a linear hash function family ${\cal F}$ and a linear code family ${\cal C}$, by using the simple fact that a kernel of a linear function is a linear space, and thus can be considered as a code.
This correspondence does not only allow us to define the code family ${\cal C}$ of a given universal hash function family ${\cal F}$, but also the dual code family ${\cal C^\perp}$ corresponding to it.
Under this setting, interestingly, a simple algebraic argument shows that the universality of ${\cal C}$ (i.e., the property of ${\cal C}$ being universal$_2$) also guarantees that of ${\cal C}^\perp$ (see Fig. \ref{fig:inclusion-universality}).
For example, (1) if ${\cal C}$ is universal$_2$, or equivalently, 1-almost universal$_2$, then ${\cal C}^\perp$ is 2-almost universal$_2$, but nevertheless, (2) for an $\varepsilon$-almost universal$_2$ code family ${\cal C}$ with $\varepsilon>1$, the dual code family ${\cal C}$ is not necessarily $\varepsilon$-almost universal$_2$, as can be seen from an explicit counterexample.
These results lead us to introduce a new class of hash functions called an $\varepsilon$-almost {\it dual} universal$_2$ hash function family, as a set of hash functions whose kernels form an $\varepsilon$-almost dual universal$_2$ code family.
This concept is indeed a generalization of the conventional universality$_2$,
since a universal$_2$ hash function family is a special case of our $\varepsilon$-almost {\it dual} universal$_2$ family.

In Section \ref{delta-biased}, we note a simple relation between our ``$\varepsilon$-almost {\it dual} universal$_2$ family'' and the concept of the ``$\delta$-biased family'', originally introduced by Dodis and Smith \cite{DS05} for correcting errors without leaking partial information.
By using this relation, we demonstrate that Renner's two-universal hashing lemma \cite[Lemma 5.4.3]{Renner} can be extended to the case where an $\varepsilon$-almost dual universal hash function family is used.
Note here that in Refs. \cite{DS05,FS08}, they did not refer this relation with privacy amplification.
This result means that the hashing lemma is valid for a broader class of hash functions than previously thought, since the conventional type of two-universal hash functions is a special case of our $\varepsilon$-almost dual universal$_2$ hash functions.

In Section \ref{s2-5}, we introduce the concept of the permuted code family, as the set of codes obtained by permuting bits of a given code $C$.
Then we show the existence of a code $C$, whose permuted family ${\cal C}_C$ is $(n+1)$-almost dual universal$_2$, with $n$ being the bit length of $C$.
The code $C$ of this type is particularly useful when the setting of our communication model is invariant under bit permutations, since the average performance of the code $C$ equals that of an $(n+1)$-almost dual universal$_2$ code family.
Due to this property, the permuted code family plays a key role in showing the existence of a deterministic hash function that works universally for different types of channels.

In Section \ref{s4}, as a preparation for later sections, we apply the results of Sections \ref{sec1} and \ref{s2-5}
to error correction. We show that a code $C\in{\cal C}$ serves as a good code when it is chosen randomly from an $\varepsilon$-almost universal$_2$ code family ${\cal C}$.

In Section \ref{s5}, we apply these results to the security proof of a QKD protocol called the Bennett-Brassard 1984 (BB84) protocol \cite{BB84}.
We use the proof technique of the Shor-Preskill--type, which reduces the security of a secret key to the error correcting property of the Calderbank-Shor-Steane (CSS) quantum error correcting code (e.g., \cite{SP00,GLLP,WMU06,H07}).
This proof technique is elegant and widely used, but also has a drawback.
That is, it requires the implementation of the classical CSS code in actual QKD systems, which can be difficult especially for large block lengths
(This is not the case for Renner's method, where universal$_2$ hash functions can be used for privacy amplification).
Our result solves this difficulty; even when one uses $\varepsilon$-almost {\it dual} universal$_2$ functions for privacy amplification, the security can be shown in the Shor-Preskill formalism.
Note here again that the conventional universal$_2$ function family is a special case of our $\varepsilon$-almost dual universal$_2$ families.

Then, in Section \ref{s6}, we apply our results on QKD to the quantum wire-tap channel.
In this model, a sender Alice has channels to two receivers, i.e.,
an authorized receiver Bob, and an unauthorized receiver Eve, often referred to as a wire-tapper.
The channels from Alice to Bob and to Eve are not necessarily restricted to any type, but we assume that they are both specified when we analyze the security.
%It is also assumed that there are two channels from Alice, i.e., one to Bob and the other to Eve.
The main issue here is to obtain an upper bound of leaked information with
with appropriate transmission rates.
The net transmission rate can be given as the information transmission rate $R'$ to Bob minus the sacrifice bit rate $R$.
The former rate can be treated in the framework of error correcting code.
The latter rate corresponds to a privacy amplification process.
%, which is a typical topic in the wire-tap channel model.

Under these settings, in Section \ref{s6}, we consider a specific type of the quantum wire-tap channel where Alice and Bob are connected by the Pauli channnel.
By applying our results on QKD to this model, we show that an $\varepsilon$-almost dual universal$_2$ function family is sufficient for removing Eve's information.
Then by using the invariance of the channel under bit permutations, we also show the existence of a deterministic hash function that works universally, that is, the hash function whose construction does not depend on the phase error probability caused by the wire-tapper.
We also clarify that
our evaluation is better than
the $\delta$-biased approach based on given \cite{DS05,FS08,cq-security}.

Finally, in Section \ref{s8},
we discuss the relation with existing results.
In Subsection \ref{s8-a},
we summarize the relation with existing results.
In Subsection \ref{s8-b},
we provide an example of an $\varepsilon$-almost universal$_2$ hash function 
family that yields insecure bits.
In Section \ref{s8-c},
we consider the case where one applies the privacy amplification after the error correction.
Then, we show that
the phase error correction approach can guarantee the strong security 
with a larger class of families of code pairs
than the $\delta$-biased approach.

\section{Dual universality of a code family}
\Label{sec1}

\subsection{Linear universal hash functions as a linear code family}
We start by reviewing the basic properties of universal$_2$ hash functions. 
Consider sets $A$ and $B$, and also a function family ${\cal F}$ consisting of functions from $A$ to $B$;
that is, ${\cal F}$ is a set of function ${\cal F}=\{f_r|r\in I\}$ with $f_r: A\to B$, where $I$ denotes a set of indices $r$ of hash functions.
Our purpose is to select $f_r$ with an equal probability and use them as a hash function, and for this purpose, we always let $|A|\ge|B|\ge2$.
We say that a function family ${\cal F}$ is {\it $\varepsilon$-almost universal$_2$} \cite{CW79,WC81}, if, for any pair of different inputs $x_1$,$x_2$, the collision probability of their outputs is upper bounded as
\begin{align}
&{\rm Pr}\left[f_r(x_1)=f_r(x_2)\right] \nonumber \\
=&\frac1{|I|}\#\left\{\,r\in I\,|\,f_r(x_1)=f_r(x_2)\,\right\}
\le \frac{\varepsilon}{|B|}.
\Label{eq:def-universal-2}
\end{align}
The parameter $\varepsilon$ appearing in (\ref{eq:def-universal-2}) is shown to be confined in the region
\begin{equation}
\varepsilon\ge\frac{|A|-|B|}{|A|-1},
\Label{eq:epsilon-lower-bound}
\end{equation}
and in particular, a function family ${\cal F}$ attaining the equality of (\ref{eq:epsilon-lower-bound}) is called an {\it optimally universal$_2$ function family} \cite{S91}.
On the other hand, a family ${\cal F}$ with $\varepsilon=1$ is simply called a {\it universal$_2$} function family.

There are three important examples of universal$_2$ hash function families:
\begin{itemize}
\item {\bf Example 1}: Toeplitz matrices (see, e.g., \cite{MNP90}).
Let $\{M_r\,|\, r\in I\}$ be a set of all $m\times n$ Toeplitz matrices.
Then for an input $x\in \FF_2^n$, the output $y\in \FF_2^m$ of function $f_r$ is given by $y=xM_r$.
\item {\bf Example 2}: Modified Toeplitz matrices (see, e.g., \cite{Hayashi2}).
Let ${\cal T}=\{T_r\,|\, r\in I\}$ be a set of all $m\times (n-m)$ Toeplitz matrix.
Then let $M_r=(T_r,I_{m})$ be an $m\times n$ matrix defined by a concatenation of $T_r$ and the $m$-dimensional identity matrix $I_m$.
For an input $x\in \FF_2^n$, the output $y\in \FF_2^m$ of function $f_r$ is given by $y=xM_r$.
%\item {\bf Example 3}: Multiplications over a finite field (see, e.g., \cite{CW79,BBCM}).
\end{itemize}
These (modified) Toeplitz matrices are particularly useful in practice, because there exists an efficient multiplication algorithm using the fast Fourier transform algorithm with complexity $O(n\log n)$ (see, e.g., \cite{MatrixTextbook}).

In this paper, we focus only on linear functions over a finite field $\FF_2$. We assume that sets $A$,$B$ are $\FF_2^n$, $\FF_2^m$ respectively with $n\ge m$, and $f_r$ are linear functions over $\FF_2$.
Note that, in this case, there is a kernel $C_r$ corresponding to each $f_r$, which is a vector space of $n-m$ dimensions or more.
Also note that, conversely, when given a vector subspace $C_r\subset\FF_2^n$ of $n-m$ dimensions or more, one can always construct a linear function
\begin{equation}
\tilde{f}_r: \FF_2^n\to \FF_2^n/C_r\cong \FF_2^{l_r}\ \ {\rm with}\ \ 
\max_r l_r= m.
\Label{eq:def-tilde-f}
\end{equation}
This means that, by considering $C_r$ as an error-correcting code\footnote{For the present, we take a standpoint that any vector subspace of $\FF_2^n$ is a code, whether or not it can actually correct errors.}, we can always identify a linear hash function $f_r$ and a error correcting code $C_r$.\footnote{Note that ${\rm dim}\,C_r={\rm dim}\,{\rm Ker}\,f_r=n-l_r$ is not a constant in general.
For example, for the function family defined by multiplication of all normal (i.e., unmodified) Toeplitz matrices of Example 1, ${\rm dim}\,C_r$ varies from $n-m$ to $n$ depending on $r\in I$.
The special case of ${\rm dim}\,C_r$ being a constant will be discussed in detail in Section \ref{subsec:surjective_linear_functions}.}

In this terminology, 
since $n- \min_r \dim C_r=m$,
the definition of $\varepsilon$-universal$_2$ function family of (\ref{eq:def-universal-2}) takes the form
\begin{equation}
\forall x\in \FF_2^n\setminus\{0\},\ \ {\rm Pr}\left[\tilde{f}_r(x)=0\right]\le 2^{-m}\varepsilon,
\end{equation}
which can further be rewritten as
\begin{equation}
\forall x\in \FF_2^n\setminus\{0\},\ \ {\rm Pr}\left[x\in C_r\right]\le 
2^{\min_r \dim C_r-n}\varepsilon.
\label{eq:universality_stated_with_code_family}
\end{equation}
This shows that the set of kernel ${\cal C}=\{C_r|r\in I\}$ contains sufficient information for determining if a function family ${\cal F}=\{f_r|r\in I\}$ is $\varepsilon$-almost universal$_2$ or not.

To see this in more detail, we give explicit constructions.
For later convenience, we denote a generating matrix of a code $C$ by $G(C)$, so that the rows of $G(C)$ are basis vectors of $C$. We also denote a parity check matrix of $C$ by $H(C)$, hence one may choose $H(C)=G(C^\perp)$.
If one wants to construct $C_r$ from $f_r$, let $x$ be a column vector, and define a linear function $f_r$ as $y=f_r(x)=M_rx$ by using an $m\times n$-matrix $M_r$.
Here $M_r$ corresponds to a parity check matrix of error-correcting code $C_r$,
and thus the row vectors of $M_r$ spans $C_r^\perp$.
Conversely, if one wants to construct a linear function $\tilde{f}_r:\FF_2^n\to\FF_2^m$ from a code $C_r$, do as follows:
First, let $l_r:=\dim C_r^\perp\le m$, and take a basis of $C_r^\perp\subset\FF_2^n$ as $\{u_1,\dots,u_{l_r}\}$, and a basis of $\FF_2^m$ as $\{v_1,\dots,v_m\}$.
Then define a matrix $\tilde{M}_r=\sum_{i=1}^{l_r}v_iu^T_i$, and let $\tilde{f}_r(x)=\tilde{M}_rx$.

It should be noted that, in fact, this construction of $\tilde{f}_r$ has an ambiguity that comes from choices of bases $\{u_i\}$ and $\{v_i\}$.
By the above procedure, even when one constructs $C_r$ from $f_r$, and then $\tilde{f}_r$ from the obtained $C_r$, $\tilde{f}_r$ and $f_r$ may not equal in general.
In this paper, however, we do not worry about this ambiguity, because
(i) the ambiguity does not affect the property of $\tilde{f}_r$ being $\varepsilon$-almost universal$_2$, and
(ii) the ambiguity is absent after all when we actually implement and operate universal hash functions for cryptographic purposes; in such cases, we never think of $C_r$ as a vector space, but rather specify matrices $M_r$ or basis sets of $C_r$ explicitly.
Note that a similar situation happens with error-correcting codes as well; i.e., it is convenient to interpret $C_r$ as a mathematical vector space when one analyzes the code theoretically, but in practice one can never implement a code as a program or a circuit without specifying the basis vectors, or equivalently, the parity check and the generating matrices.

\subsection{Dual universality of a code family}
\Label{sec:approx-duality}From these arguments, we define the universality of error-correcting codes as follows.
\begin{Dfn}
We define the minimum (respectively, maximum) dimension of a code family ${\cal C}=\{C_r|r\in I\}$ as  $t_{\min}:=\min_{r\in I}\dim C_r=\min_{r\in I}n-l_r$ (respectively, $t_{\max}:=\max_{r\in I}\dim C_r=\max_{r\in I}n-l_r$).
\end{Dfn}
\begin{Dfn}
We define the dual code family ${\cal C}^\perp$ of a given linear code family ${\cal C}=\{C_r|r\in I\}$ as the set of all dual codes of $C_r$. That is, ${\cal C}^\perp=\{C_r^\perp|r\in I\}$.
\end{Dfn}
\begin{Dfn}
We say that a linear code family ${\cal C}=\{\,C_r\subset\FF_2^n\,|\,r\in I\,\}$ of minimum dimension $t_{\min}$ is an $\varepsilon$-almost universal$_2$ code family
of minimum dimension $t_{\min}$ 
, if the following condition is satisfied
\begin{equation}
\forall x\in \FF_2^n\setminus\{0\},\ \ {\rm Pr}\left[x\in C_r\right]\le 2^{t_{\min}-n}\varepsilon.
\Label{eq:C-r-upperbound-1}
\end{equation}
Relaxing Condition \ref{eq:C-r-upperbound-1}, 
we say that a linear code family ${\cal C}=\{\,C_r\subset\FF_2^n\,|\,r\in I\,\}$ of maximum dimension $t_{\max}$ is an $\varepsilon$-almost universal$_2$ code family
of maximum dimension $t_{\min}$ 
, if the following condition is satisfied
\begin{equation}
\forall x\in \FF_2^n\setminus\{0\},\ \ {\rm Pr}\left[x\in C_r\right]\le 2^{t_{\max}-n}\varepsilon.
\Label{eq:C-r-upperbound-2}
\end{equation}
\end{Dfn}

As in the case of a universal$_2$ function family, $\varepsilon$ is bounded from below by (\ref{eq:epsilon-lower-bound}) as $\varepsilon\ge(2^n-2^{n-t})/(2^n-1)$. For the case where $\varepsilon$ achieves this minimum, we say that ${\cal C}$ is optimally universal$_2$. Similarly, if $\varepsilon=1$, we call ${\cal C}$ a universal$_2$ code family.

We also introduce the notion of dual universality as follows.
\begin{Dfn}
We say that a code family ${\cal C}$ is $\varepsilon$-almost dual universal$_2$
of maximum (minimum) dimension $t$
, if the dual family ${\cal C}^\perp$ is $\varepsilon$-almost universal$_2$
of minimum (maximum) dimension $t$.
\end{Dfn}
Hence, accordingly,
\begin{Dfn}
\label{Dfn:def_dual_function_family}
A linear function family ${\cal F}=\{f_r|r\in I\}$ is $\varepsilon$-almost dual universal$_2$, if the kernels $C_r$ of $f_r$ form an $\varepsilon$-almost dual universal$_2$ code family.
\end{Dfn}
An explicit example of a dual universal$_2$ function family (with $\varepsilon=1$) can be given by the modified Toeplitz matrices 
(Example 2)
mentioned earlier \cite{H07}, i.e., a concatenation $(X, I)$ of the Toeplitz matrix $X$ and the identity matrix $I$.
This example is particularly useful in practice because it is both universal$_2$ and dual universal$_2$ (c.f., Fig. \ref{fig:inclusion-universality}), and also because there exists an efficient algorithm with complexity $O(n\log n)$.

Indeed, since Condition (\ref{eq:C-r-upperbound-1}) coincides with
(\ref{eq:universality_stated_with_code_family}),
it seems it is enough to use only Condition (\ref{eq:C-r-upperbound-1}).
In the case of Example 1,
a large part of 
Kernels of $M_r$ takes
their dimension to be the maximum dimension 
$n-m$ of the code family.
Then, 
Kernels of $M_r$ forms 
an $\varepsilon$-almost universal$_2$ code family 
of maximum dimension $n-m$ with $\varepsilon=1$.

However, when we consider 
$\varepsilon$-almost dual universal$_2$ family of hash functions,
our situation becomes more complex.
In the case of Example 1,
a large part of 
dual codes of Kernels of $M_r$ takes
their dimension to be the minimum dimension 
$m$ of the code family.
In this case, the vector $x$ belongs to the 
dual code of Kernel of $M_r$ 
if and only if
$x$ can be written as a linear combination of row vectors of $M_r$.
Hence, we can show that
\begin{eqnarray*}
{\rm Pr}[x \in (\Ker M_r)^{\perp}]
\le 2^{m-n},
\end{eqnarray*}
which implies that
$\{M_r|r \in I\}$ is  
an $\varepsilon$-almost $_2$ dual universal$_2$ 
code family function family with $\varepsilon=1$.
Hence, Condition (\ref{eq:C-r-upperbound-2}) is essential for 
$\varepsilon$-almost dual universality$_2$.

With these preliminaries, we can present the following main theorem of this section:
\begin{Thm}
\Label{thm:almost-universal}
Given an $\varepsilon$-almost universal$_2$ code family ${\cal C}$ of minimum dimension $t$, the dual code family ${\cal C}^\perp$ is a $2(1-2^{t-n}\varepsilon)+(\varepsilon-1)2^t$-almost universal$_2$ code family with maximum dimension $n-t$. That is, for $\forall x\in \FF_2^n\setminus\{0\}$, the dual code family ${\cal C}^\perp$ satisfies
\begin{equation}
{\rm Pr}\left[x\in C_r^\perp\right]\le 
(1-2^{t-n}\varepsilon)2^{-t+1}+\varepsilon-1.
\Label{eq:inequality-thm1}
\end{equation}
In other words, the code family ${\cal C}$ is also $2(1-2^{t-n}\varepsilon)+(\varepsilon-1)2^t$-almost dual universal$_2$.
\end{Thm}

\begin{proof}
For $x,y\in \FF_2^n$, let
\begin{eqnarray}
p_x&:=&{\rm Pr}\left[x\in C_r^\perp\right],\\
V_x&:=&\left\{y \in \FF_2^n| (x,y)=0 \right\}=\{x,0\}^\perp,
%r_y&:=&2^{t-n}\varepsilon-{\rm Pr}\left[y\in C_r\right]\ge 0,
\end{eqnarray}
where $(x,y)$ denotes the inner product of $x,y$.
Since $\#(V_x\setminus \{0\})=2^{n-1}-1$, 
\begin{align}
& 2^{t-n}\varepsilon (2^{n-1}-1)
=\sum_{y \in V_x\setminus \{0\}} 
2^{t-n}\varepsilon \nonumber \\
\ge &  
\sum_{y \in V_x\setminus \{0\}} {\rm Pr}\left[y\in C_r\right] .
\Label{Ha46}
\end{align}
Now, (i) If $x \in C_r^\perp$,
it means that $C_r \subset V_x$, and we have $\dim (C_r\cap V_x)=\dim C_r\ge t$.
Hence it follows that
$\# (C_r \cap V_x \setminus \{0\})=\# (C_r \setminus \{0\})\ge 2^t-1$.
On the other hand, (ii) If $x \notin C_r^\perp$,
we have $\dim (C_r\cap V_x)\ge t-1$, and thus 
$\# (C_r \cap V_x \setminus \{0\}) \ge 2^{t-1}-1$.
Because
$\sum_{y \in V_x\setminus \{0\}} {\rm Pr}\left[y\in C_r\right]$
is equal to the average of the number of 
$\# (C_r \cap V_x \setminus \{0\})$,
relations (i) and (ii) yields
%can bound the sum of $\Pr[y\in C_r]$ from above by $p_x$ as
\begin{align}
\sum_{y \in V_x\setminus \{0\}} {\rm Pr}\left[y\in C_r\right] 
\ge & p_x(2^t-1)+ (1-p_x)(2^{t-1}-1) \nonumber \\
 =  & 2^{t-1}+ p_x 2^{t-1}-1\Label{Ha45}.
\end{align}
Combining (\ref{Ha46}) and (\ref{Ha45}),
we have $2^{t-n}(2^{n-1}-1)\varepsilon \ge 2^{t-1}+ p_x 2^{t-1}-1$,
which leads to inequality (\ref{eq:inequality-thm1}).
\end{proof}

\begin{Thm}
\Label{thm:equality}
Inequality (\ref{eq:inequality-thm1}) of Theorem \ref{thm:almost-universal} is tight.
That is, 
for an integer $t \le n$,
an element $x\in \FF_2^n \setminus \{0\}$, and 
a positive real number $\varepsilon\le \frac{2-2^{1-t}}{1-2^{1-n}}$,
there exists an $\varepsilon$-almost universal$_2$ code family ${\cal C}$ with minimum dimension $t$ 
satisfying the equality of (\ref{eq:inequality-thm1}).
\end{Thm}

In the above theorem,
the real number 
$\varepsilon =\frac{2-2^{1-t}}{1-2^{1-n}}$
is the maximum number satisfying 
$(1-2^{t-n}\varepsilon)2^{-t+1}+\varepsilon-1 \le 1$.

\begin{proof}
Fix $x\in \FF_2^n$. Then define a code family ${\cal A}=\{A_r\}$ in $\FF_2^n$ as follows.
Choose randomly an $t$-dimensional subspace of $V_x=\{y \in \FF_2^n| (x,y)=0 \}$.
That is, select $t$ linearly independent elements from $V_x$ randomly, and let them span a subspace $A_r$. Then one has:
\begin{equation}
y\in V_x\setminus\{0\},\ \ \Pr\left[y\in A_r\right]=\frac{2^t-1}{2^{n-1}-1}.
\Label{eq:a-ineq}
\end{equation}

We also define another code family ${\cal B}=\{B_r\}$ as follows.
First choose a $t-1$-dimensional subspace of $V_x$ randomly, and then include an additional basis element $z\not\in V_x$ to it, so that they form an $t$-dimensional subspace in total. Then the following inequalities hold:
\begin{eqnarray}
y\in V_x\setminus\{0\},&\ &\Pr\left[y\in B_r\right]=\frac{2^{t-1}-1}{2^{n-1}-1},\Label{eq:b-ineq1}\\
y\not\in V_x,&\ &\Pr\left[y\in B_r\right]=2^{t-n}.\Label{eq:b-ineq2}
\end{eqnarray}

Finally, define a code family ${\cal C}=\{C_r\}$ by combining ${\cal A}$ with probability $p$, and ${\cal B}$ with probability $1-p$, where $p$ is defined by
\begin{equation}
p:=\left(1-2^{t-n}\varepsilon\right)2^{-t+1}+\varepsilon-1.
\end{equation}
One may wonder that this construction using probability $p$ deviates from our definition of universal$_2$ code family that each element $C_r$ is chosen with the uniform probability.
One way to cure this problem is to include multiple copies of ${\cal A}$ and ${\cal B}$ in  ${\cal C}$.
For example, if $p=a/b$ with $a,b\in\NN$, then construct ${\cal C}$ as a combination of $a$ copies of ${\cal A}$ and $b-a$ copies of ${\cal B}$.

\par From (\ref{eq:a-ineq}), (\ref{eq:b-ineq1}), and (\ref{eq:b-ineq2}), it is straightforward to see that ${\cal C}$ is $\varepsilon$-almost universal$_2$.
Also note, since $x\in C_r^\perp$ holds only when ${\cal A}$ is chosen, we have
\begin{equation}
\Pr\left[x\in C_r^\perp\right] = p.
\end{equation}
Hence, ${\cal C}$ indeed attains the equality of (\ref{eq:inequality-thm1}).
\end{proof}

We give some useful examples of Theorems \ref{thm:almost-universal} and \ref{thm:equality}.
We apply these results to several communication models in later sections.
\begin{Crl}\Label{Hacor}
The following relations hold for a code family ${\cal C}$ and the dual family ${\cal C}^\perp$:
\begin{enumerate}
\item If ${\cal C}$ is optimally universal$_2$, ${\cal C}^\perp$ is also optimally universal$_2$.
In other words, an optimally universal$_2$ family ${\cal C}$ is also optimally dual universal$_2$.
\item If ${\cal C}$ is universal$_2$ (i.e., $1$-almost universal$_2$), ${\cal C}^\perp$ is $2$-almost universal$_2$.
In other words, a universal$_2$ family ${\cal C}$ is also $2$-almost dual universal$_2$.
\item For $\varepsilon>1$, however, an $\varepsilon$-almost universal$_2$ family ${\cal C}$  is not necessarily $\varepsilon'$-almost dual universal$_2$.
That is, there is an example of an $\varepsilon$-almost universal$_2$ family ${\cal C}$ with $\max_x\Pr[x\in C_r^\perp]=1$.
\end{enumerate}
\end{Crl}
\begin{proof}
Items 1 and 2 are obvious.
For item 3, choose $\varepsilon$ so that the right hand side of (\ref{eq:inequality-thm1}) equals 1.
\end{proof}

\subsection{Case of sujective linear function family}
\label{subsec:surjective_linear_functions}
Some linear function families ${\cal F}=\{f_r:\FF_2^n\to \FF_2^m\,|\,r\in I\}$ consist only of surjective functions $f_r$, i.e., functions $f_r$ satisfying ${\rm Im}\,f_r=\FF_2^m$ for all $r\in I$.
In this case, it is straightforward to show that the dimension of the corresponding code family ${\cal C}=\{C_r\,|\,r\in I\}$ is constant: ${\rm dim}\,C_r= {\rm dim\ Ker}\,f_r=n-m$.

The goal of this subsection is to demonstrate that, for these particular families, the definitions and the theorems of the previous section concerning dual universal$_2$ functions can be greatly simplified.
We take this particular case, because we believe that it provides an intuitive picture on results of the previous subsections;
e.g., the dual universality can be discussed directly without mentioning the corresponding code family ${\cal C}$.
However, at the same time, it should also be noted that there are many useful examples of {\it non-surjective} hash function families including Toeplitz matrices of Example 1.
Hence in the rest of paper, we do {\it not} restrict ourselves to surjective function family;
instead we consider general linear hash functions as defined in the previous subsection.

We begin by defining duality of surjective function families:
\begin{Dfn}
\label{def:dual_functions}
Given two surjective linear functions $f:\FF_2^n\to \FF_2^m$ and $g:\FF_2^n\to \FF_2^{n-m}$, we say that $f$ and $g$ are dual functions if ${\rm Ker}\, f= ({\rm Ker}\, g)^\perp$, or equivalently, if ${\rm Ker}\, g= ({\rm Ker}\, f)^\perp$.
\end{Dfn}

We note that a similar definition can be found in Ref. \cite{Renes10}.
It is straightforward to generalize this notion to function families:
\begin{Dfn}
Given two function families consisting only of surjective functions and having the same index $r\in I$,
\begin{eqnarray*}
{\cal F}&=&\{f_r:\FF_2^n\to \FF_2^{m}\,|\,r\in I\},\\
{\cal G}&=&\{g_r:\FF_2^n\to \FF_2^{n-m}\,|\,r\in I\},
\end{eqnarray*}
we say that ${\cal F}$ and ${\cal G}$ are dual families, if $f_r$ and $g_r$ are dual functions for all $r\in I$.
\end{Dfn}

Recall from Definition \ref{Dfn:def_dual_function_family} that a function family ${\cal F}$ is $\varepsilon$-almost dual universal$_2$ iff the corresponding code family ${\cal C}=\{C_r^\perp\,|\,r\in I\}=\{({\rm Ker}\,f_r)^\perp\,|\,r\in I\}$ is $\varepsilon$-almost universal$_2$.
For a dual pair of surjective families ${\cal F}$ and ${\cal G}$, this is equivalent to the condition that ${\cal C}^\perp=\{{\rm Ker}\,g_r|r\in I\}$ is $\varepsilon$-almost universal$_2$.
Then by noting the definition of universality$_2$ given in (\ref{eq:universality_stated_with_code_family}), we can redefine the universality of surjective families in a simpler way:
\begin{Dfn}
A surjective function family ${\cal F}$ is $\varepsilon$-almost universal$_2$, iff its dual function family ${\cal G}$ is $\varepsilon$-almost  universal$_2$.
\end{Dfn}

Theorem \ref{thm:almost-universal} can also be simplified as:
\begin{Crl}
%\Label{thm:almost-universal}
If a surjective function family ${\cal F}=\{f_r:\FF_2^n\to \FF_2^{m}\,|\,r\in I\}$ is $\varepsilon$-almost universal$_2$, then its dual function family ${\cal G}=\{g_r:\FF_2^n\to \FF_2^{n-m}\,|\,r\in I\}$ is $2(1-2^{-m}\varepsilon)+(\varepsilon-1)2^{n-m}$-almost  universal$_2$
\end{Crl}

It is convenient to consider these statements in terms of matrices.
Take an arbitrary pair of surjective linear functions, $f:\FF_2^n\to\FF_2^m$ and $g_r:\FF_2^n\to\FF_2^{n-m}$.
Then $f$ can be written as a matrix multiplication $y=xM$, with input $x$ and output $y$, and with $M$ being an $m\times n$ matrix.
Similarly, $g$ can also be expressed as $y=xN$ with an $(n-m)\times n$ matrix $N$.
Since the row vectors of $M$, $N$ form a basis of  $({\rm Ker}\,f)^\perp$, $({\rm Ker}\,g)^\perp$, respectively, we conclude that $f$ and $g$ are dual functions iff $MN^T=0$.

Hence, a straightforward way of constructing a pair ${\cal F}$, ${\cal G}$ of dual family is as follows:
First choose a code family ${\cal C}=\{C_r^\perp\,|\,r\in I\}$ of a fixed dimension.
Then define functions $f_r$ by $y=xG(C_r)$ with $G(C_r)$ being the generating matrix of $C_r$,
and $g_r$ by $y=xH(C_r)$ with $H(C_r)$ being the parity check matrix.
In this case, if ${\cal F}$ is $\varepsilon$-universal$_2$, then one can guarantee that ${\cal G}$ is $\varepsilon'$-universal$_2$, with $\varepsilon$ and $\varepsilon'$ related as in Theorem \ref{thm:almost-universal}.

One useful example that fits this construction is the family of all modified Toeplitz matrices, given as Example 2.
In this case, the presence of the identity matrix $I_{m}$ maximizes rank $M_r$ and guarantees the surjectivity of the corresponding linear function.
It is easy to see that the dual families are defined by $N_r=(I_{n-m},T_r^T)$, which is another class of modified Toeplitz matrices (note $M_rN_r^T=0$).

Still, it should also be noted that there are many useful examples of {\it non-surjective} hash function family.
For example, for the normal Toeplitz matrices of Example 1, the rank of $T_r$ ranges from zero to $m$ depending on $r$ (consider the case where its rows are periodic).
Hence in the rest of this paper, we do not restrict ourselves to surjective function family;
instead we consider general linear hash functions as defined in the previous subsection.

\subsection{Generalization to subcode, extended code, and code pair families}
\Label{sec:generalize-subcode}
For the application to quantum key distribution, it is convenient to generalize the concept of a universal$_2$ code family to those ${\cal C}=\{C_{2,r}\}$ consisting solely of extended codes of $C_1$.
\begin{Dfn}
Let $C_1\subset \FF_2^n$ be a fixed $m$-dimensional code.
A code family ${\cal C}_2=\{C_{2,r}\,|\,r\in I\}$ is called an extended code family of $C_1$,
if each $C_{2,r}$ is an extended code of $C_1$, i.e., $\forall r\in I,$ $C_1\subset C_{2,r}$.
An extended code family ${\cal C}$ of $C_1$
is called an $\varepsilon$-almost universal$_2$ extended code family of $C_1$
with minimum (or maximum) dimension $t$, 
if 
\[\forall x\in \FF_2^n \setminus C_1,
\ \Pr\left[x\in C_{2,r}\right]
=
\Pr\left[ [x] \subset C_{2,r} \right]
\le2^{t-n}\varepsilon,\]
where $[x]$ denotes the coset with the representative $x$ in 
$\FF_2^n / C_1$.
\end{Dfn}
By considering a universality of a dual code family of such extended code family, we are naturally led to the following definition of universal$_2$ subcode families.
\begin{Dfn}
Let $C_1\subset \FF_2^n$ be a fixed $m$-dimensional code.
A code family ${\cal C}_2=\{C_{2,r}\,|\,r\in I\}$ is called a
subcode family of $C_1$,
if each $C_{2,r}$ is a subcode of $C_1$, i.e., $\forall r\in I,$ 
$ C_{2,r}\subset C_1$.
A subcode family ${\cal C}_2$ 
of $C_1$
is called an $\varepsilon$-almost universal$_2$ subcode family of $C_1$ 
with minimum (or maximum) dimension $t$, if 
\[\forall x\in C_1\setminus\{0\},\ 
\Pr\left[x\in C_{2,r}\right]
\le 2^{t-m}\varepsilon.\]
\end{Dfn}

\begin{Dfn}
Let $C_1\subset \FF_2^n$ be a fixed $m$-dimensional code.
A code family ${\cal C}_2=\{C_{2,r}\,|\,r\in I\}$ is called a
subcode family of $C_1$,
if each $C_{2,r}$ is a subcode of $C_1$, i.e., $\forall r\in I,$ 
$ C_{2,r}\subset C_1$.
A subcode family ${\cal C}_2$ 
of $C_1$
is called an $\varepsilon$-almost dual universal$_2$ subcode family of $C_1$ 
with minimum (or maximum) dimension $t$, if 
the extended code family ${\cal C}^{\perp}$ of $C_1^{\perp}$
is an $\varepsilon$-almost universal$_2$ extended code family of $C_1^{\perp}$
with maximum (or minimum) dimension $n-t$.
Similarly,
an extended code family ${\cal C}$ of $C_1$
is called an $\varepsilon$-almost dual universal$_2$ extended code family of $C_1$
with minimum (or maximum) dimension $t$, 
if
a subcode family ${\cal C}_2^{\perp}$ 
of $C_1^{\perp}$
is called an $\varepsilon$-almost universal$_2$ subcode family of $C_1$ 
with maximum (or minimum) dimension $n-t$.
\end{Dfn}

One explicit construction of ${\cal C}_2$ is to first let ${\cal D}=\{D_r\in\FF_2^m|r\in I\}$ be a universal$_2$ code family with minimum dimension $t$, and then define generating matrix of $C_{2,r}\in{\cal C}_2$ by $G(C_{2,r}):=G(D_r)G(C_1)$.
For these types of codes as well, we can prove a theorem similar to Theorems
\ref{thm:almost-universal} and \ref{thm:equality}.
\begin{Thm}
\Label{thm:subcode-dual}
Let $C_1\subset \FF_2^n$ be a fixed $m$-dimensional code, 
and ${\cal C}_2$ be an $\varepsilon$-almost universal$_2$ subcode 
family ${\cal C}_2$ of $C_1$ with minimum dimension $t \le m$.
Then the dual code family ${\cal C}_2^\perp$ is a $2(1-2^{t-m}\varepsilon)+(\varepsilon-1)2^t$-almost universal$_2$ extended code (subcode) family of $C_1^\perp$ with maximum dimension $n-t$.
That is,
\begin{align}
\forall x \in \FF_2 \setminus C_1^{\perp},\ 
\Pr\left[x\in C_{2,r}^{\perp}\right]
\le 
(1-2^{t-m}\varepsilon)2^{-t+1}+\varepsilon-1 .
\Label{Ha30}
\end{align}
In other words, the subcode family ${\cal C}_2$ is also a $2(1-2^{t-m}\varepsilon)+(\varepsilon-1)2^t$
-almost dual universal$_2$ extended code family of $C_1$.

Moreover,
for an integer $t \le m$,
an element $x \in \FF_2 \setminus C_1^{\perp}$,
and
a positive real number $\varepsilon
\le \frac{2-2^{1-t}}{1-2^{1-m}}$,
there exists an $\varepsilon$-almost universal$_2$ subcode 
family ${\cal C}_2$ of $C_1$ 
with minimum dimension $t$ 
satisfying the equality of (\ref{Ha30}).
\end{Thm}

\begin{proof} 
For an $\varepsilon$-almost universal$_2$ subcode 
(extended code) family ${\cal C}_2$ of $C_1$,
the equivalence relations 
$C_1 \cong \FF_2^n/C_1^{\perp} \cong \FF_2^{m}$
hold.
The proofs of the above theorems with $\FF_2^m$
can be applied to this theorem.
\end{proof}

\begin{Thm}
\Label{thm:excode-dual}
Let $C_1\subset \FF_2^n$ be a fixed $m$-dimensional code, 
and ${\cal C}_2$ be an $\varepsilon$-almost universal$_2$ extended code
family ${\cal C}_2$ of $C_1$ with minimum dimension $t \ge m$.
Then the dual code family ${\cal C}_2^\perp$ is a $2(1-2^{t-n}\varepsilon)+(\varepsilon-1)2^{t-m}$-almost 
universal$_2$ subcode family of $C_1^\perp$ with maximum dimension $n-t$.
That is,
\begin{align}
\Pr\left[x\in C_{2,r}^{\perp}\right] 
\le 
(1-2^{t-n}\varepsilon)2^{-t+m+1}+\varepsilon-1
\Label{Ha31}
\end{align}
for $
\forall 
x \in C_1^{\perp}\setminus\{0\}$.
In other words, the extended code family ${\cal C}_2$ is also a $2(1-2^{t-n}\varepsilon)+(\varepsilon-1)2^{t-m}$
-almost dual universal$_2$ subcode family of $C_1$.

Furthermore,
for an integer $m \le t \le n$,
an element $x \in C_1^{\perp}\setminus\{0\}$,
and
a positive real number $\varepsilon
\le \frac{2-2^{1-t+m}}{1-2^{1-n+m}}$,
there exists an $\varepsilon$-almost universal$_2$ 
extended code
family ${\cal C}_2$ of $C_1$ 
with minimum dimension $t$ 
satisfying the equality of (\ref{Ha31}).
\end{Thm}

\begin{proof}
Similarly,
for an $\varepsilon$-almost universal$_2$ 
extended code family ${\cal C}_2$ of $C_1$,
the equivalence relations 
$\FF_2^n/C_1 \cong C_1^{\perp} \cong \FF_2^{n-m}$
hold.
Under this equivalence, 
$C_{2,r}/C_1$ can be regarded as subspace of $\FF_2^{n-m}$
with the minimum dimension $t-m$.
The proofs of the above theorems with $\FF_2^{n-m}$
and the minimum dimension $t-m$
can be applied to this theorem.
\end{proof}

Furthermore, when the code $C_1$ is randomly chosen,
the concept of 
an extended code family $\{C_{2,r}\}_r$ can be 
generalized to the following way.
In this case, 
we define the property ``$\varepsilon$-almost universal$_2$''
for 
a family of a pair of codes $\{ C_{1,r}\subset C_{2,r} \}_r$.
\begin{Dfn}
A family of a pair of codes $\{ C_{1,r}\subset C_{2,r} \}_r$
is called an {\it $\varepsilon$-almost universal$_2$ code pair family}
with minimum (or maximum) dimension $t$
when it 
satisfies the condition 
\begin{align*}
t = \min_{r} \dim C_{2,r} (\max_{r} \dim C_{2,r})\\
\forall x\in \FF_2^n \setminus \{0\} ,
\ \Pr\left[x\in C_{2,r}\setminus C_{1,r} \right]
\le 2^{t-n}\varepsilon.
\end{align*}
\end{Dfn}

Since
any $\varepsilon$-almost universal$_2$ extended code family $\{C_{2,r}\}_r$ of the code $C_1$
gives 
an 
$\varepsilon$-almost universal$_2$ code pair family
$\{ C_{1}\subset C_{2,r} \}_r$,
the concept ``$\varepsilon$-almost universal$_2$ code pair family''
is generalization of
``$\varepsilon$-almost universal$_2$ extended code family''.

Considering the dual codes, we obtain the following definition.
\begin{Dfn}
a family of a pair of codes $\{ C_{1,r}\subset C_{2,r} \}_r$
is called an {\it $\varepsilon$-almost dual universal$_2$ pair family}
with maximum (or minimum) dimension $t$
if
a family of a pair of codes $\{ C_{2,r}^{\perp}\subset C_{1,r}^{\perp} \}_r$
is an $\varepsilon$-almost universal$_2$ code pair family
with minimum (or maximum) dimension $n-t$.
\end{Dfn}
Since
any $\varepsilon$-almost dual universal$_2$ subcode family $\{C_{2,r}\}_r$ of the code $C_1$
gives 
an 
$\varepsilon$-almost dual universal$_2$ code pair family
$\{ C_{2,r}\subset C_{1} \}_r$,
the concept ``$\varepsilon$-almost dual universal$_2$ code pair family''
is generalization of
``$\varepsilon$-almost dual universal$_2$ subcode family''.

\section{The $\delta$-biased family}\Label{delta-biased}
Next, according to Dodis and Smith\cite{DS05},
we introduce $\delta$-biased family of random variables $\{W_{r}\}$.
%For a random variable $D$ on $\FF_2^n$ and $x\in \FF_2^n$, we define
%\begin{align}
%\biase_x(D):= \rE_{D} (-1)^{x\cdot D}.
%\end{align}
For a given $\delta>0$,
a family of random variables $\{W_{r}\}$ on $\FF_2^n$
is called {\it $\delta$-biased}
when the inequality
\begin{align}
\rE_{r} (\rE_{W_{r}} (-1)^{x\cdot W_{r}})^2 \le \delta^2
\end{align}
holds for any $x\in \FF_2^n$, $x\ne 0$.

We denote the random variable subject to the uniform distribution on a code $C\in \FF_2^n$ by $W_C$.
Then,
\begin{align}
\rE_{W_C} (-1)^{x\cdot W_{C}}
=
\left\{
\begin{array}{ll}
0 & \hbox{ if } x \notin C^{\perp} \\
1 & \hbox{ if } x \in C^{\perp} .
\end{array}
\right.
\end{align}
Using this relation, we obtain the following lemma. 
\begin{Lmm}\Label{Lem6-0}
When a code family ${\cal C}=\{C_{r}\subset \FF_2^n\}_r$ with minimum dimension $n-m$ is $\varepsilon$-almost dual universal,
the family of random variables $\{W_{C_{r}}\}$ on $\FF_2^n$
is $\sqrt{\varepsilon 2^{-m}}$-biased.
\end{Lmm}

Hence an $\varepsilon$-almost dual universal$_2$ code family
yields a $\delta$-biased family.
%$\sqrt{\varepsilon 2^{-m}}$-biased family.
For 
a partially eavesdropped random viable $A$ and
a $\delta$-biased family of random variables $\{W_r\}_r$ that is independent from Eve's random variable,
Dodis and Smith \cite{DS05}
proposed the protocol 
\begin{align}
(A,W_r)\mapsto A+W_r \Label{2-2-1}
\end{align}
for error correction with leaking partial information.
%[{\cite[Lemma 4]{DS05}}]\Label{Lem6-1}
%When the $\delta$-biased family is produced by an $\varepsilon$-almost dual universal$_2$ code family,
%the protocol proposed by Dodis and Smith\cite{DS05}
%for error correction with leaking partial information.
%can be written by Fig. \ref{fig:two-protocols1}.
In order to evaluate the leaked information of this protocol,
they showed the classical version of the following lemma (Lemma \ref{Lem6-1-q}).
Fehr and Schaffner \cite{FS08} 
extended it to the quantum case
in order to discuss the property of the protocol against a quantum attacker.

In this section, with the help of Lemmas \ref{Lem6-0} and \ref{Lem6-1-q}, 
we evaluate the leaked information 
after the privacy amplification by an $\varepsilon$-almost dual universal$_2$ code family.

Given 
a classical-quantum state $\rho^{A,E}=
\sum_{a}P^A(a)| a \rangle\langle a|\otimes \rho_{a}^E$ 
on ${\cal H}_A \otimes {\cal H}_E$,
and
a normalized state $\sigma^E$ on ${\cal H}_E$,
Renner \cite{Renner} defines
\begin{align}
d_1(A :E| \rho^{A,E} )
&:=
\|\rho^{A,E}-\rho_{\mix}^A\otimes \rho^E \|_1 ,
\end{align}
and
\begin{align*}
& d_2(A :E| \rho^{A,E} \|\sigma^E) \\
:= &
2^{-H_2(A|E|\rho^{A,E}\|\sigma^E)}
-\frac{1}{|{\cal A}|}\Tr ((\sigma^E )^{-1/4} \rho^E (\sigma^E )^{-1/4} )^2
\\
& H_2(A|E|\rho^{A,E}\|\sigma^E) \\
:= & -\log_2 
\Tr 
((I \otimes \sigma^{E})^{-1/4}
\rho^{A,E}
(I \otimes \sigma^{E})^{-1/4})^2 \\
& H_{\min}(A|E|\rho^{A,E}\|\sigma^E) \\
:= & -\log_2 \|  
(I \otimes \sigma^{E})^{-1/2}
\rho^{A,E}
(I \otimes \sigma^{E})^{-1/2} \| .
%H_{\min}^{\epsilon} (A|E|\rho^{A,E}\|\sigma^E)
%&:=
%\max_{\rho':\|\rho^{A,E} -\rho'\|_1 \le \epsilon} 
%H_{\min}(A|E|\rho'\|\sigma).
\end{align*}
As relations among these quantities,
Renner \cite[Lemma 5.2.3]{Renner} shows 
\begin{align}
d_1(A :E| \rho^{A,E} )
\le &
\sqrt{|{\cal A}|}
\sqrt{d_2(A :E| \rho^{A,E} \|\sigma^E)}
\Label{1-3-1} \\
H_2(A|E|\rho\|\sigma)
\ge &
H_{\min}(A|E|\rho\|\sigma) 
\Label{1-3-2} 
\end{align}

For a distribution $P^{W}$ on ${\cal A}$,
we define another 
classical-quantum state 
$\rho^{A,E}* P^{W}:=
\sum_{w}P^{W}(w) \sum_{a}P^A(a)| a+w\rangle\langle a+w|\otimes \rho_{a}^E$, which describes the output state of the protocol 
(\ref{2-2-1}).
Then, the following lemma holds.
\begin{Lmm}[{\cite[Theorem 3.2]{FS08}}]\Label{Lem6-1-q}
For any c-q sub-state $\rho^{A,E}$ on ${\cal H}_A \otimes {\cal H}_E$
and 
any state $\sigma^E$ on ${\cal H}_E$,
a $\delta$-biased
family of random variables $\{W_{r}\}$ on ${\cal A}$
satisfies
\begin{align}
\rE_{r} d_2(A :E| \rho^{A,E}* P^{W_{r}} \|\sigma^E)
\le
\delta^2
2^{-H_{2}(A|E|\rho^{A,E}\|\sigma^E)}.
\end{align}
\end{Lmm}

Based on the above lemma, we can evaluate
the average performance of the privacy amplification by
$\varepsilon$-almost dual universal$_2$ code family
as follows.

\begin{Lmm}\Label{Lem6-3-q}
Given a classical-quantum state $\rho^{A,E}$ 
on ${\cal H}_A \otimes {\cal H}_E$
and a state $\sigma^E$ on ${\cal H}_E$.
When $\{C_{r}\}$ is a $\varepsilon$-almost dual universal$_2$ 
code family with minimum dimension $m$,
the family of hash functions $\{f_{C_{r}}\}_{r}$ 
satisfies
\begin{align}
\rE_{r} d_2(f_{C_{r}}(A):E| \rho^{A,E}\|\sigma^E)
\le
\varepsilon 
%\Tr \rho
2^{-{H}_{2}(A|E|\rho^{A,E}\|\sigma^E )}.
\Label{12-5-9-q}
\end{align}
That is, 
any $\varepsilon$-almost dual universal$_2$ 
hash function family $\{f_{r}\}_{r}$
satisfies the above inequality.
\end{Lmm}

Using (\ref{1-3-1}) and (\ref{1-3-2}),
we obtain
\begin{align}
\rE_{r} d_1(f_{C_{r}}(A):E| \rho^{A,E})
\le &
\varepsilon 
2^{\frac{n-m}{2}-\frac{1}{2}{H}_{2}(A|E|\rho^{A,E}\|\sigma^E )}\nonumber \\
\le &
\varepsilon 
2^{\frac{n-m}{2}-\frac{1}{2}{H}_{\min}(A|E|\rho^{A,E}\|\sigma^E )}.
\end{align}

Thus we have obtained the $\varepsilon$-almost dual universal version of Theorem 5.5.1 of Renner \cite{Renner}.
Hence, the two-universal hashing lemma and other results as given in Renner \cite{Renner} can be generalized to our $\varepsilon$-almost dual universal hash functions.
Note here that, as we have shown in Section \ref{sec1}, the conventional universal$_2$ function family is a special case of our $\varepsilon$-almost dual universal$_2$ families.
In the following, 
in order to distinguish the method given in Sections \ref{s5} and \ref{s6},
we call this approach to 
the privacy amplification by
$\varepsilon$-almost dual universal$_2$ code family,
the $\delta$-biased approach.

%So, we obtain $\epsilon$-almost universal version of Theorem 5.5.1 of Renner \cite{Renner}.
%Hence, we can expect that similar analysis as given in Renner \cite{Renner} is possible for $\epsilon$-almost universal hash functions.

%applying smoothing method,
%we obtain
%\begin{align}
%\rE_{r} d_1(f_{C_{r}}(A):E| \rho^{A,E})
%\le
%2 \epsilon' 
%+
%\epsilon 
%2^{\frac{n-m}{2}-\frac{1}{2}{H}_{\min}^{\epsilon'}(A|E|\rho^{A,E}\|\sigma^E )}.
%\end{align}

\begin{proof}
Due to Lemma \ref{Lem6-1-q},
we obtain
\begin{align}
\rE_{r} d_2(A :E|
\rho^{A,E}* P^{W_{C_{r}}} \|\sigma^E)
\le
\varepsilon 2^{-m}
2^{-{H}_{2}(A|E|\rho\|\sigma)}.\Label{12-18-2}
\end{align}

Now, we focus on the relation ${\cal A} \cong {\cal A}/C \times C\cong f_{C} \times C$
for any code $C$.
That is, any $a\in {\cal A}$ can be uniquely specified by a coset element $[a]=a+C$ and a codeword $w\in C$.
We regard $[a]$ as the hash value $f(a)$ of $a$.
Then, for $P_W(w)=2^{-m}$, we obtain
\begin{eqnarray*}
\lefteqn{\rho^{A,E}* P^{W}
=
\sum_{w \in C}
2^{-m}
\sum_{a}P^A(a)| a+w\rangle_A \langle a+w|\otimes \rho_{a}^E}\\
&=&
\sum_{w \in C}
2^{-m}|w\rangle_{W} \langle w|
\otimes 
\sum_{[a] \in {\cal A}/C}
P^{A}([a])| [a]\rangle_{F} \langle [a]|\otimes \rho_{[a]}^E\\
&= & \sum_{w \in C}
2^{-m}|w\rangle_W \langle w|
\otimes 
\rho^{f_{C}(A),E}.
%\label{eq:proof_lemma_dual_leftover_1}
\end{eqnarray*}
In the second and the third lines, we used a new set of basis such that $|a\rangle_A=|w\rangle_W\otimes|[a]\rangle_F$.
Probability $P^A([a])$ denotes that of a coset element $[a]$ occurring: $P^A([a]):=\sum_{w\in C}P^A(a+w)$,
and similarly, $\rho_{[a]}^E$ the mixed state corresponding to $[a]$, i.e., $\rho_{[a]}^E:=\sum_{w\in C}\rho^E_{a+w}$.
Then by the definition of $d_2$, we have
\begin{align*}
& d_2(A :E|
\rho^{A,E}* P^{W_{C}} \|\sigma^E) \\
=&
2^{-m}
d_2(f_{C}(A) :E|\rho^{f_{C}(A),E}\|\sigma^E) \\
=&
2^{-m}
d_2(f_{C}(A) :E|\rho^{A,E}\|\sigma^E).
\end{align*}
Therefore, (\ref{12-18-2}) implies
\begin{align*}
\rE_{r} 2^{-m}
d_2(f_{C_{r}}(A) :E|\rho^{A,E}\|\sigma^E)
\le
\varepsilon 2^{-m}
2^{-{H}_{2}(A|E|\rho^{A,E}\|\sigma^E )},
\end{align*}
which implies (\ref{12-5-9-q}).
\end{proof}

\begin{rem}
One might think that
the concept of ``$\varepsilon$-almost dual universal$_2$ hash function  family"
is not needed because of the correspondence between an $\varepsilon$-almost dual universal$_2$ hash function family and a $\delta$-biased family given in Lemma \ref{Lem6-0}.
However, if we replace the terminology 
``$\varepsilon$-almost dual universal$_2$ hash function family"
by the terminology ``$\delta$-biased family",
we make a serious confusion
by the following reasons.
\begin{enumerate}
\item The concept of the ``$\delta$-biased family" is defined for a family of random variables
while the concept of the ``$\varepsilon$-almost dual universal$_2$ hash function family"
is defined for a family of hash functions.
It is confusing to use the terminology ``$\delta$-biased family'' for describing a family of hash functions.

\item
The correspondence holds only when a $\delta$-biased family is given as the uniform distribution on a code.
Other $\delta$-biased families do not necessarily have such correspondence.

\item
If we study hash functions only in terms of the concept of the $\delta$-biased family, their relation with universal$_2$ hash functions family becomes obscure.
\end{enumerate}
\end{rem}

\section{Permuted code family}\Label{s2-5}
In some applications, our setting is invariant under permutations of the order of bits
in $\FF^n_2$.
For example, in wire-tap channels which we consider in later sections, independent and identically distributed (i.i.d.) channels are assumed and thus the protocol is invariant under permutations of bits.
Then a code $C \subset \FF^n_2$ has the same performance as any bit-permuted code of $C$.

In order to formulate such situations, we introduce 
the {\it permuted code family} of a code $C$
as a code family consisting of bit-permuted codes of $C$
\begin{eqnarray}
{\cal C}_C:=\{ \sigma(C) |\sigma \in S_n\}.
\end{eqnarray}
Here $S_n$ denotes the symmetric group of degree $n$,
and $\sigma(i)=j$ means that $\sigma\in S_n$ maps $i$ to $j$, where $i,j\in\{1,\dots,n\}$.
The code $\sigma(C)$ is the one obtained by permuting bits of $C$ by a permutation $\sigma$;
if $x=(x_1,\dots,x_n)\in C$, then $x^\sigma:=(x_{\sigma(1)},\dots,x_{\sigma(n)})\in \sigma(C)$.

In what follows, we denote the distribution of the Hamming weight $k$ of codewords in $C$ by
$\Pr_C$; that is, the number of codewords with weight $k$ contained in $C$ is $|C|\Pr_C(k)$.
%Similarly, the weight distribution $\Pr_{\cal C}$ of a code family ${\cal C}$ is obtained by averaging $\Pr_C$ over $C\in {\cal C}$ with an equal probability.
In order to characterize the permuted code family ${\cal C}_C$,
when the dimension of a code $C$ is $t$,
we define
\begin{eqnarray}
\varepsilon_k(C):=&
\frac{|C|\Pr_C(k)}{{n \choose k}} 2^{-t+n}
= 
\frac{2^n \Pr_C(k)}{{n \choose k}} \\
\varepsilon(C):=&
\max_{1 \le k \le n}\varepsilon_k(C).
\label{eq:epsilon_permuted_code}
\end{eqnarray}

\begin{Lmm}\Label{ht1}
The permuted code family ${\cal C}_C$ is $\varepsilon(C)$-almost universal$_2$ code family.
\end{Lmm}

\begin{proof}
%Let $x\in \FF_2^n$ be an element with weight $1\le k\le n$.
Any code $C' \in  {\cal C}_C$ has the weight distribution $\Pr_C$.
By averaging them over all $C' \in  {\cal C}_C$, we see that code family ${\cal C}$ also has the weight distribution $\Pr_C$. That is, a code $C'\in{\cal C}_C$ contains $2^t \Pr_C(k)$ elements of weight $k$ on average.
On the other hand, the number of elements $x\in \FF_2^n$ with weight $k$ is ${n \choose k}$, and due to the symmetry of ${\cal C}_C$ under bit permutations, each of them is contained in some $C' \in  {\cal C}_C$ with the same probability.
Thus, 
an element $x\in \FF_2^n$ with weight $k$ belongs to the code $C' \in  {\cal C}_C$ with the probability $\frac{|C| \Pr_C(k)}{{n \choose k}}$.
By taking the maximum with respect to $k$, 
we can show that
any element $x\in \FF_2^n$ belongs to the code $C' \in  {\cal C}_C$ with the probability $\varepsilon(C) 2^{t-n}$.
Hence, we obtain the desired argument.
\end{proof}

\begin{Thm}\Label{ht2}
For any $1\le t \le n$,
there exists a $t$-dimensional code $C \in \FF_2^n$ such that
$\varepsilon(C)
\le n+1$.
\end{Thm}

\begin{proof}
Let ${\cal C}$ be a universal$_2$ code family.
Then,
$\rE \varepsilon_k(C)\le 1$. The Markov inequality yields
\begin{eqnarray}
\Pr\{
\varepsilon_k(C) \ge n+1 \}
\le \frac{1}{n+1},
\end{eqnarray}
and thus
\begin{align*}
& \Pr
\{ \varepsilon_1(C) < n+1, \ldots,  \varepsilon_n(C) < n+1\}^c \\
=&
\Pr
\bigcup_{1\le k \le n}
\{\varepsilon_k(C) \ge n+1 \}
\le \frac{n}{n+1}.
\end{align*}
Hence, there exists a code $C$ such that
\begin{eqnarray}
\varepsilon_k(C) < n+1
\end{eqnarray}
for $k=1, \ldots, n$.
\end{proof}

Combining Lemma \ref{ht1} and Theorem \ref{ht2}, 
we obtain the following proposition.
\begin{proposition}\Label{2-7-1}
For any $1 \le t \le n$,
there exists a $t$-dimensional code $C$ such that 
the permuted code family ${\cal C}_C$ is $n+1$-almost universal$_2$.
\end{proposition}

Indeed, Shulman et al. \cite{SF99}
discussed the average of decoding error probability under the permuted code family.
However, 
we do not consider the average of decoding error probability, here.
We show the relation with the concept of $\varepsilon$-almost universal$_2$
while they did not treat the relation with the concept.

Similarly, we can define the permuted code pair family for a given pair of codes $C_2\subset C_1$
as the family of code pairs ${\cal C}_{C_2\subset C_1}:=\{ \sigma(C_2) \subset \sigma(C_1) |\sigma \in S_n\}$. 
We define $\varepsilon(C_1/C_2):=\max_{1\le k \le n}\varepsilon_k(C_1)-\varepsilon_k(C_2)\frac{|C_2|}{|C_1|}$.
As a generalization of Lemma \ref{ht1},
we obtain the following lemma.

\begin{Lmm}\Label{ht1-2}
The permuted code pair family ${\cal C}_{C_1/C_2}$ is $\varepsilon(C_1/C_2)$-almost universal$_2$ code pair family.
\end{Lmm}

This lemma can be shown by the same discussion as the proof of Lemma \ref{ht2}.
Furthermore, we can show the following theorem.

\begin{Thm}\Label{ht2-2}
For any $t \le n$ and a code $C_2$,
there exists a $t$-dimensional code $C_1 \in \FF_2^n$ such that
$C_2 \subset C_1$ and 
$\varepsilon(C_1/C_2)\le n+1$.
\end{Thm}
This theorem can be shown in the same way as Theorem \ref{ht2} 
by choosing the code $C_1$ from a 
universal$_2$ extended code family of $C_2$.

Combining Lemma \ref{ht1-2} and Theorem \ref{ht2-2}, 
we obtain the following proposition.
\begin{proposition}\Label{2-7-2}
For any $1 \le t \le n$ and a code $C_2$,
there exists a $t$-dimensional extended code $C_1$ of $C_2$
such that 
the permuted code pair family ${\cal C}_{C_1/C_2}$ 
is an $n+1$-almost universal$_2$ code pair family.
\end{proposition}

Considering the dual codes, we obtain the following 
proposition.

\begin{proposition}\Label{2-7-3}
For any $1 \le t \le n$ and a code $C_2$,
there exists a $t$-dimensional subcode $C_1$ of $C_2$
such that 
the permuted code pair family ${\cal C}_{C_2/C_1}$ 
is an $n+1$-almost dual universal$_2$ code pair family.
\end{proposition}

Proposition \ref{2-7-3} can be shown by
substituting $C_2^{\perp}$ and $C_1^{\perp}$ into 
$C_2$ and $C_1$ in Proposition \ref{2-7-2}.
In later sections, we use these results for showing the existence of deterministic hash function that work universally for 
quantum wire-tap channels.

\section{Application to error correcting codes}\Label{s4}
In this section, as a preliminary for later section, we apply the results of Section \ref{sec1} to error correction.
We use a code $C\in{\cal C}$ chosen randomly from an $\varepsilon$-almost universal$_2$ code family ${\cal C}$ for error correction, and show that it indeed serves as a good code.
As previous work, for example, Brassard and Salvail applied universal$_2$ codes in the context of information reconciliation (Ref. \cite{BS93}, Theorem 6).
Muramatsu and Miyake have also studied a similar problem using a somewhat generalized definition of universal hash functions \cite{MM1}.
Here we present a much simpler evaluation by employing a more restrictive condition for the family of codes than \cite{MM1}.

We consider a noisy channel with the additive noise, and denote the probability that the noise $x\in \FF_2^n$ occurs by $P^X(x)$.
We also denote by $\hat{P}^X(k)$ the probability that an error with the Hamming weight $k$ occurs.
In this channel, the sender Alice uses an $\varepsilon$-almost universal$_2$ code family as error correcting codes.
% for noisy channels.
The receiver Bob applies the maximum likelihood decoder to his bits.
%, which is equivalent to the minimum Hamming distance decoder in this channel.
In order to evaluate the performance of the decoder, we focus on the decoding error probability, i.e., the probability that the decoder makes a wrong guess.
We denote this probability for a fixed code $C$ by $P_e(C)$. From now on, 
we often treat a code $C$ as a random variable that is randomly chosen with the equal probability from the $\varepsilon$-almost universal$_2$ code family ${\cal C}$.
For example, we denote the expectation of variable $A$ with respect to the random variable $C$ as $\rE_{C\in {\cal C}} A$.
In this notation, the main purpose of this section is to evaluate $\rE_{C\in {\cal C}} P_e(C)$, i.e., the average of $P_e(C)$ when $C$ is randomly chosen from ${\cal C}$.

First, for the sake of simplicity, we evaluate performance of the minimum Hamming distance decoder.
Note that the decoding error probability of this decoder, $E_{C\in{\cal C}}P_{\rm hd}(C)$, can be used as an upper bound on $E_{C\in{\cal C}}P_e(C)$, since the maximum-likelihood decoder provides the minimum decoding error probability $P_e(C)$.
We assume that our $\varepsilon$-almost universal$_2$ code family ${\cal C}$ has the maximum dimension $t_{\max}$;
hence the decoder outputs $t_{\max}$ bits, and the code rate is $R=t_{\max}/n$.
Now suppose that a bit flip $x$ of Hamming weight $k$ occurs in the channel (i.e., an input $w$ is mapped to $w+x$).
In this case, success and failure of the decode by the minimum decoding
is written by $P_{\rm hd}(x,C)$.
That is, the success (the failure) is denoted by $P_{\rm hd}(x,C)=0$
($P_{\rm hd}(x,C)=1$).
Then the decoder fails
if there exists another code element $y \in C$ with Hamming weight $\le k$;
in other words, if $\{y\in \FF_2^n\, :\,|y|\le k, y\ne x\} \cap (C \setminus \{0\}) \neq \emptyset$.
Then, 
\begin{align}
P_{\rm hd}(x,C) 
=&
1[\{y\in \FF_2^n\, :\,|y|\le k, y\ne x\} \cap (C \setminus \{0\}) \neq \emptyset]
\nonumber  \\
\le & 
\sum_{y:|y|\le k, y\ne x} 1[y\in (C \setminus \{0\})],
\label{H-11-7-1}
\end{align}
where $1[A]$ is the indicator function defined 
to be $1$ when $A$ is valid and to be $0$ otherwise.
For a fixed element $y$, 
due to Condition (\ref{eq:C-r-upperbound-2}), 
any $\varepsilon$-almost universal$_2$ code family ${\cal C}$ satisfies
\begin{eqnarray}
E_{C\in{\cal C}}\,1[y\in C ]
=
{\rm Pr}[y\in C ]
\le 2^{t_{\max}-n}\varepsilon
\label{H-11-7-2}
\end{eqnarray}
for $y \ne 0$.
When averaged over $C\in{\cal C}$, 
combining (\ref{H-11-7-1}) and (\ref{H-11-7-2}),
we can evaluate the average probability of
$P_{\rm hd}(x,C)$
\begin{eqnarray*}
E_{C\in{\cal C}}P_{\rm hd}(x,C)& \le &E_{C\in{\cal C}}\sum_{y:|y|\le k, y\ne x}
1[y\in (C \setminus \{0\})]\\
&\le &E_{C\in{\cal C}} \sum_{y:|y|\le k}
1[y\in (C \setminus \{0\})]\\
&=& E_{C\in{\cal C}}\sum_{y:|y|\le k, y\ne 0}1[y\in C ]\\
&=& \sum_{y:|y|\le k, y\ne0}E_{C\in{\cal C}}\,1[y\in C ]\\
&\le& \sum_{y:|y|\le k, y\ne0} 2^{t_{\max}-n}\varepsilon\\
&\le& 2^{n h\left(\min \{k/n,1/2\} \right)}2^{t_{\max}-n}\varepsilon,
\end{eqnarray*}
where 
the final inequality follows from the fact that $\sum_{i=0}^k{n\choose i}\le 2^{n h\left(\min \{k/n,1/2\} \right)}$ (see, e.g., Lemma 4.2.2 of \cite{Justesen}).
%Since this bound depends only on $k$ rather than $x$, we are justified in writing $E_{C\in{\cal C}}P_{\rm hd}(k;C)$ instead of $E_{C\in{\cal C}}P_{\rm hd}(x,C)$.
Also by noting the obvious bound $E_{C\in{\cal C}}P_{\rm hd}(x;C)\le 1$, we have
\begin{equation}
E_{C\in{\cal C}}P_{\rm hd}(x;C)\le
\varepsilon
2^{-n\left[1-h\left( \min\{|x|/n,1/2\} \right)-R\right]_+}
\label{eq:decoding_error_bound_with_memory_fixed_k}
\end{equation}
for $\varepsilon \ge 1$, where $[a]_+:=\max\{a,0\}$ for $a\in \RR$.

Since the behavior of the minimum Hamming distance decoder is independent of parameter $k$, the bound (\ref{eq:decoding_error_bound_with_memory_fixed_k}) can easily be generalized to the case in the following way
where a weight distribution $\hat{P}^X(k)$ of errors is given.
%By averaging over $k$, 
\begin{eqnarray}
%\lefteqn{
&&\rE_{C\in {\cal C}} P_e(C) 
= 
\rE_{C\in {\cal C}} 
\sum_{x\ne 0\in \FF_2^n}
P^X(x) 
P_{\rm hd}(x;C)\nonumber\\
&=& 
\sum_{x\ne 0\in \FF_2^n}
P^X(x) 
\rE_{C\in {\cal C}} 
P_{\rm hd}(x;C)\nonumber\\
&\le &
\varepsilon
\sum_{x\ne 0 \in \FF_2^n}
P^X(x) 
2^{-n\left[1-h\left(\min \{|x| /n,1/2\} \right)-R\right]_+} \nonumber\\
&=& 
\varepsilon
\sum_{k=1}^n
\hat{P}^X(k) 
2^{-n\left[1-h\left(\min \{k/n,1/2\} \right)-R\right]_+}.
\Label{Ha20}
\end{eqnarray}
As to the asymptotic behavior, one can easily see that,
when the probability $\hat{P}^X \{k| 1-h\left(\min \{k/n,1/2\} \right) > R+\delta \}$ approaches $1$ for sufficiently small $\delta>0$,
the right hand side of (\ref{Ha20}) converges to zero.
We note that Inequality (\ref{Ha20}) is used in Ref. \cite{HT12} to prove the security of the BB84 protocol for the case of finite key lengths.

\begin{rem}
The essential point for the above evaluation for 
$\rE_{C\in {\cal C}} P_e(C) $ is 
the exchange of the orders of 
$\sum_{x\ne 0}$ and 
$E_{C\in{\cal C}}$.
For a fixed error $x$, 
the $\varepsilon$-almost universality$_2$ 
guarantees the evaluation of the average 
$E_{C\in{\cal C}}P_{\rm hd}(x;C)$
as (\ref{eq:decoding_error_bound_with_memory_fixed_k}).
If we fix a code $C$, we cannot obtain a similar evaluation. 
\end{rem}

Next we consider the cases of finite $n$.
In this case it is not easy to calculate similar bounds, hence we further assume that the channel is memoryless.
That is,
the probability distribution $P^X$ of errors $x$ is assumed to be the binary distribution with probability $p$.
In this channel, 
when $p$ is less than $1/2$,
the maximum-likelihood decoder is equivalent to 
the minimum Hamming distance decoder.
In this case, by modifying Gallager's bound for the random coding \cite{Gal}, we can obtain the following simple bound.
\begin{Thm}\Label{thm3}
When $P^X(x)$ is given as the $n$-th independent and identical distribution of the distribution
$(1-p,p)$, then the average decoding error probability of error correction using an $\varepsilon$-almost universal$_2$ code family ${\cal C}$ with maximum dimension $t_{\max}=nR$ satisfies
\begin{align}
\rE_{C\in {\cal C}} P_e(C) \le 
 \min_{0 \le s \le 1}
\varepsilon^s
2^{-n\left[-sR+E_0(s,p)\right]},
\Label{eq:Ecc-bound1}
\end{align}
where
\begin{equation}
E_0(s,p):=s-\log_2\left[p^{\frac1{1+s}}+\left(1-p\right)^{\frac1{1+s}}\right]^{1+s}.
\label{eq:def_E0}
\end{equation}
\end{Thm}

This theorem is shown in Appendix \ref{s7}.
The function $E_0(s,p)$ defined in (\ref{eq:def_E0}) is in fact the specialized form of Gallager's $E_0(s,\mbox{\boldmath$p$})$ for the binary symmetric channel and the uniform input distribution \cite{Gal}.
Hence by using the method of \cite{Gal}, the right hand side of (\ref{eq:Ecc-bound1}) can be used to evaluate the exponential decreasing rate of $\rE_{C\in {\cal C}} P_e(C)$ with respect to $n$ as follows.
\begin{Crl}\label{crl:ERp}
Under the same conditions as Theorem \ref{thm3}, $\rE_{C\in {\cal C}} P_e(C)$ can be bounded from above as
\begin{equation}
\rE_{C\in {\cal C}} P_e(C)\le 2^{-nE(R,p)}\max\{\varepsilon,1\}
\label{eq:upper_rECPeC_epsilon_ge1}
\end{equation}
where $E(R,p)$ is Gallager's reliability function
\begin{equation}
E(R,p):=\max_{0\le s\le1}-sR +E_0(s,p).
\label{eq:def_E_Rp}
\end{equation}
In particular, $E(R,p)$ is strictly positive for $R < 1-h(p)$.
\end{Crl}

{\it Proof of Corollary \ref{crl:ERp}:}
The first half of the corollary is obvious.
Denote the argument of the maximum by $E_R(s,p):=-sR+E_0(s,p)$.
Then $E_R(0,p)=0$, and $\left.\frac{\partial}{\partial s}E_R(s,p)\right|_{s=0}=1-h(p)-R>0$ if $R < 1-h(p)$.
Hence $E_R(s,p)$ attains its positive maximum value at $s\in(0,1]$.
(Also see Ref. \cite{Gal}.)
\endproof

The exponential decreasing rate $E(R,p)$ of (\ref{eq:upper_rECPeC_epsilon_ge1}) can also be verified from (\ref{Ha20}) by using the type method \cite{CKbook}
when $p \le 1/2$.
For this purpose, we introduce the divergence function
$d(q\|p):=q\log\frac{q}{p}+(1-q)\log\frac{1-q}{1-p}$.
Since $\hat{P}^X(k)\le 2^{-nd(q\|p)}$ with $q=k/n$ for the binary symmetric channel \cite{CKbook}
and
$\sum_{k= \lceil n/2 \rceil }^n \hat{P}^X(k) 
\le 2^{-nd(1/2 \|p)}$,
the right hand side of (\ref{Ha20}) can be evaluated as
\begin{align}
& \varepsilon
\sum_{k=0}^n
\hat{P}^X(k) 
2^{-n\left[1-h\left( \min\{ k/n, 1/2\} \right)-R\right]_+} \nonumber \\
\le &
\varepsilon
(2^{-nd(1/2 \|p)}\nonumber \\
&+
\lfloor n/2+1 \rfloor
\max_{0 \le k \le n/2}
\hat{P}^X(k) 
2^{-n\left[1-h\left(k/n\right)-R\right]_+} 
)
\nonumber \\
\le &
\lfloor n/2+2 \rfloor
\varepsilon
\max_{0\le q\le 1/2}
2^{-n\left([1-h(q)-R]_+ +d(q\|p)\right)} \nonumber \\
=&
\lfloor n/2+2 \rfloor
\varepsilon
2^{-n \min_{0\le q\le 1/2}[1-h(q)-R]_+ +d(q\|p)}.
\Label{Ha20-4}
\end{align}

One can see that the exponential decreasing rate of (\ref{Ha20-4}) indeed equals $E(R,p)$ by using the relation
\begin{align}
\min_{0\le q \le 1/2}[1-h(q)-R]_+ +d(q\|p)
= \max_{0 \le s \le 1/2}-sR +E_0(s,p).
\Label{11-23-1}
\end{align}
The proof of this relation is given, e.g., in Csisz\'{a}r-K\"{o}rner \cite{CKbook} in a more general form.
However, since a simpler proof of (\ref{11-23-1}) can be given by using the property of additive channels, we reproduce it in Appendix \ref{app:Proof_min_DQP} for readers' convenience.

Now, we consider the case where the sender and the receiver use a fixed $t$-dimensional code $C$ that satisfies the condition of Theorem \ref{ht2}, i.e., a code $C$ whose permuted code family ${\cal C}_C$ is $(n+1)$-almost universal$_2$.
If the error distribution $P^X$ is permutation invariant, e.g., if the channel is binary symmetric, we have $P_e(C)=P_e\left(\sigma(C)\right)$ for any permutation $\sigma \in S_n$, which implies that $P_e(C)= \rE_{\sigma \in S_n} P_e(\sigma(C))$.
In other words, one may evaluate $P_e(C)$ as if the code family ${\cal C}_C$ were actually used.
Thus, by applying (\ref{eq:Ecc-bound1}) and by noting $n+1>1$, we obtain the inequality
\begin{align}
P_e(C) \le
(n+1) 2^{-nE(R,p)}
\Label{Ha20-2}
\end{align}
with $R=t/n$.
Note that the code $C$ satisfies this inequalities for any $p$.

In the rest of this section, we show that the above results also hold for the case where the information is encoded by the coset $C_1/C_2$
of two given codes $C_1$ and $C_2$ satisfying $C_2\subset C_1\subset \FF_2^n$.
These codes are used for constructions of the quantum Calderbank-Shor-Steane (CSS) codes,
and for this reason, they are often called the classical CSS codes.
In this section, we restrict ourselves to the following type of classical communication.
A message to be sent is a coset $[x]\in C_1/C_2$, and when the sender wants to send $[x]$, she chooses an element randomly from the set $x+C_2$ with the equal probability and sends it.
On the receiver's side, Bob first applies the maximum likelihood decoder of $C_1$ on the received sequence and obtains an element $y\in C_1$. Then, he obtains a coset $[y] \in C_1/C_2$ as the final decoded message.
We denote the decoding error probability of this decoder by $P_e(C_1/C_2)$.

We assume that the subcode $C_2$ is fixed, and 
the larger code $C_1$ is randomly chosen with the equal probability
from the $\varepsilon$-almost universal$_2$ extended code family ${\cal C}$ of $C_2$ with maximum dimension $t_{\max}$.
Again, the purpose of the following discussion is to evaluate $\rE_{C_1\in {\cal C}} P_e(C_1/C_2)$.
By a similar argument as above,
when the bit flip error occurs on $k$ bits in the noisy channel,
we can show that $\rE_{C_1\in {\cal C}} P_e(C_1/C_2)$ is less than $\min\{2^{nh\left(\min \{k/n,1/2\}\right)} \varepsilon 2^{t_{\max}-n},1\}
\le
\varepsilon
2^{-n[1-h\left(\min \{k/n,1/2\}\right)-R]_+}$, $R=t_{\max}/n$
for $\varepsilon \ge 1$.
Thus, for any weight distribution $\hat{P}^X$ of errors, we have
%this average of the decoding error probability under the above decoder is less than
\begin{align}
\rE_{C_1\in {\cal C}} P_e(C_1/C_2)
\le 
\varepsilon
\sum_{k=0}^n 
\hat{P}^X(k)
2^{-n\left[1-h\left(\min \{k/n,1/2\}\right)-R\right]_+}.
\Label{Ha17}
\end{align}

If we further assume the channel is memoryless,
%, i.e.,
%the probability distribution $P^X$ concerning the error $x$ is assumed to be the binary distribution with probability $p$.
%Since the maximum likelihood decoder realizes the minimum decoding error probability, this average of the decoding error probability under the maximum likelihood decoder is also less than the above value.
as a generalization of Theorem \ref{thm3} and Corollary \ref{crl:ERp}, we have the following.
\begin{Thm}\Label{thm4}
When $P^X(x)$ is given as the $n$-th independent and identical distribution of the distribution
$(1-p,p)$, then an $\varepsilon$-almost universal$_2$ extended code family
${\cal C}$ of $C_2$ with the maximum dimension $t_{\max}=nR$ satisfies
\begin{equation}
\rE_{C_1\in {\cal C}} 
P_e(C_1/C_2)
\le 
 \min_{0 \le s \le 1}
\varepsilon^s
2^{-n\left[-sR+E_0(s,p)\right]}.
\Label{Ha7-1}
\end{equation}
and thus
\begin{equation}
\rE_{C_1\in {\cal C}} 
P_e(C_1/C_2)
\le
2^{-nE(R,p)}\max\{\varepsilon,1\}.
\Label{Ha7-1a}
\end{equation}
Further, 
the above inequalities are valid even with
an $\varepsilon$-almost universal$_2$
extended code pair family $\{C_1 \subset C_2\}$.
\end{Thm}
This theorem is also shown in Appendix \ref{s7} in a way similar to Theorem \ref{thm3}.

Finally, for a given code $C_2$, we can choose another fixed code $C_1$ satisfying the condition of Theorem \ref{ht2-2}, i.e., 
$C_2 \subset C_1$ and $\varepsilon(C_1/C_2)\le n+1$.
We then assume that the sender and the receiver use this fixed pair for error correction.
If the distribution $P^X$ is permutation invariant, we have $P_e(C_1/C_2)=P_e(\sigma(C_1)/\sigma(C_2))$ for any permutation $\sigma \in S_n$, which implies that
$P_e(C_1/C_2)= \rE_{\sigma \in S_n} P_e(\sigma(C_1)/\sigma(C_2))$.
Thus one may evaluate $P_e(C_1/C_2)$ as if the $n+1$-almost universal$_2$ permuted extended code pair family 
${\cal C}_{C_2\subset C_1}$ were actually used.
Applying (\ref{Ha7-1}), we obtain the inequality
\begin{align}
P_e(C_1/C_2)
 \le 
(n+1)
2^{-nE(R,p)}.
\Label{Ha7-2}
\end{align}
Note that the code $C_1$ satisfies this inequality for any $p$.

\section{Quantum key distribution}\Label{s5}
In this section, 
we show the strong security when an $\varepsilon$-almost dual universal hash 
function family is applied in the quantum key distribution (QKD).
For this purpose, 
we apply the results of previous sections 
to the phase error correction in the security proof of quantum key distribution (QKD).
Hence, we call this approach the phase error correction approach. 

In QKD, Alice and Bob need to perform a key distillation protocol to generate a secret key from the sifted key that they obtained as a result of the quantum communication.
We consider the following type of the BB84 protocol using a function family ${\cal F}=\{f_r:\FF_2^m\to\FF_2^l|r\in I\}$ for privacy amplification.
\begin{center}
\fbox{
\begin{minipage}{0.9\linewidth}
BB84 protocol using universal hash function family:
\begin{enumerate}
\item Alice and Bob establish sifted keys, and estimate the bit error rate by the usual procedure of the BB84 protocol, such as the one given in \cite{SP00}.
That is,
\begin{enumerate} 
\item Alice sends Bob qubit states chosen randomly out of $\{|0_z\rangle,|1_z\rangle,|0_x\rangle,|1_x\rangle\}$.
\item Bob receives and measures them with randomly chosen bases $\{z,x\}$.
\item By using the authenticated public channel, Bob announces his measurement bases for all qubits, and they keep only the bits for which they chose the same basis.
\item They reveal randomly sampled bits over the public channel, and  calculate the estimated bit error rate. If the rate is too high, they abort the protocol.
\end{enumerate}
As a result, Alice and Bob obtains sifted key $k_A, k_B\in\FF_2^n$, respectively.
\item Alice picks a random number $r_A\in\FF_2^l$, and announces $v=k_A\oplus G(C_1)r_A^T$,  with $\oplus$ denoting XOR.
\item Bob calculates $R_B=k_B\oplus v$ and by correcting its errors using $C_1$, he obtains $R'_B\in C_1$.
Then he calculate raw bit $r_B\in\FF_2^l$ satisfying $R_B=G(C_1)r_B^T$.
(Thus $r_A=r_B$ with high probability).
\item Alice selects a linear universal$_2$ function $f_r:\FF_2^m\to\FF_2^l$ randomly and announces it to Bob. Then they calculate secret keys $s_A=f_r(r_A)$ and $s_B=f_r(r_B)$.
\end{enumerate}
\end{minipage}
}
\end{center}
By using the widely used proof technique due to Shor and Preskill \cite{SP00,GLLP,WMU06,H07}, the unconditonal security of this protocol has been shown for the case where ${\cal F}$ consists of the completely random linear functions \cite{WMU06,H07}.
On the other hand, by using the quantum de Finneti representation theorem, Renner proved the unconditional security of the BB84 protocol using universal$_2$ hash functions for privacy amplification \cite{Renner}.
In this section, we present a security proof of the Shor-Preskill--type that holds with a weaker condition on ${\cal F}$, i.e., with ${\cal F}$ being an $\varepsilon$-almost dual universal$_2$ family.
Note that the condition on ${\cal F}$ is indeed relaxed, since, as shown in Sec. \ref{sec1}, the universal$_2$ function family is a limited case of $\varepsilon$-almost dual universal$_2$ families.

Note also that our method has an extra advantage that, unlike in \cite{Renner}, Alice and Bob do not need to perform random permutations of the sifted key bits.
Conversely, if the random permutation is already implemented in one's QKD system, 
or if the channel is permutation invariant, our hash function can be replaced 
by the one using the deterministic code obtained in Theorem \ref{ht2-2},
since the permuted codes of this code pair 
form an $(n+1)$-almost dual universal$_2$ subcode pair family.

For showing the security, it is convenient to rewrite the protocol in terms of the classical CSS code as follows.

\begin{center}
\fbox{
\begin{minipage}{0.9\linewidth}
BB84 protocol using code family ${\cal C}_2$:
\begin{enumerate}
\item Alice and Bob establish sifted keys $k_A, k_B\in\FF_2^n$ by the same procedure as  in the above protocol.
\item Alice picks $R_A\in C_1$ randomly and sends $v=k_A\oplus R_A$ to Bob over the public channel.
\item Bob calculates $R_B=v\oplus k_B$, and by correcting its errors using $C_1$, he obtains $R_B'\in C_1$. (Thus $R_A=R_B'$ with high probability.)
\item Alice selects code $C_{2,r}$ randomly and announces it to Bob.
They both obtain secret keys as cosets of $C_{2,r}$, i.e., $S_A=R_A+C_{2,r}$, $S_B=R_B'+C_{2,r}$.
\end{enumerate}
\end{minipage}
}
\end{center}

For the sake of simplicity, we will restrict ourselves to this protocol for the rest of this section.
We begin by reviewing some of the known results and clarify notations.
Assume that the quantum channel between Alice and Bob is given by an arbitrary quantum operation $\Lambda$, and thus the sifted key is affected by $\Lambda$.
As discussed in \cite{Hayashi3,H07}, since the above type of the BB84 protocol is invariant under twirling of qubits, without loss of generality, one may consider the Pauli channel $\Lambda_t$ obtained by twirling the original channel $\Lambda$.
The Pauli channel $\Lambda_t$ can generally be described by the joint probability distribution $P^{XZ}$ of phase error and bit error (in this section, we call an error in the $x$ basis the phase error, and in the $z$ basis the bit error).
That is, $\Lambda_t$ transforms an $n$-qubit state $\rho$ to
\begin{equation}
\Lambda_t(\rho)=\sum_{x,z\in \FF_2^n} P^{XZ}(x,z)
\mbox{\boldmath$Z$}^x\mbox{\boldmath$X$}^z
\rho
\left(\mbox{\boldmath$Z$}^x\mbox{\boldmath$X$}^z\right)^\dagger,
\end{equation}
where
\begin{eqnarray*}
\mbox{\boldmath$Z$}^x:&=&\sigma_z^{x_1}\otimes\cdots\otimes\sigma_z^{x_n},\\
\mbox{\boldmath$X$}^z:&=&\sigma_x^{z_1}\otimes\cdots\otimes\sigma_x^{z_n}
\end{eqnarray*}
with $\sigma_x$ and $\sigma_z$ being the Pauli matrices, and $x=(x_1,\dots,x_n)$, $z=(z_1,\dots,z_n)\in \{0,1\}^n$.
We denote the marginal distribution of phase error by $P^X(x)=\sum_{z\in \FF_2^n}P^{XZ}(x,z)$.
As in the previous section, $\hat{P}^X(k)$ denotes the distribution of the Hamming weight $k$ of $x$ obeying $P^X(x)$.

Next, before considering the secret key, we evaluate the security of the sifted key $v$ as an illustration.
% quantitatively.
The result here will also be used in later sections on wire-tap channels and randomness extraction.
Let $\rho_{A,E}$ be Alice's and Eve's total system when the when the first step of the protocol (i.e., the quantum communication part) is finished.
If one employs the security criteria that takes into account the universal composability \cite{Renner}, the security of the sifted key can be evaluated by Eve's distinguishability $\left\|\rho_{A,E}-\rho_A \otimes \rho_E \right\|_1$, with $\rho_A:=\Tr_E \rho_{A,E}$ and $\rho_E:=\Tr_A \rho_{A,E}$%
\footnote{
Recall that, in our protocol, Alice is assumed to choose her sifted key uniformly. 
Hence $\rho_{A,E}$ can generally be described as $\rho_{A,E}:=\sum_{v_1,\ldots,v_n}\frac{1}{2^n}
|v_1,\ldots,v_n\rangle \langle v_1,\ldots,v_n| \otimes \rho_E(v_1,\ldots,v_n)$,
where $\rho_E(v_1,\ldots,v_n)$ denotes Eve's density matrix when Alice's sifted key is $v=(v_1,\ldots,v_n)$.
In this case, $\Tr_E \rho_{A,E}=\sum_{v_1,\ldots,v_n}\frac{1}{2^n}
|v_1,\ldots,v_n\rangle \langle v_1,\ldots,v_n|$ gives the fully mixed state.}.
Alternatively, one may evaluate the security by Eve's Holevo information $\chi:= \Tr \rho_{A,E} (\log \rho_{A,E}-\log \rho_A \otimes \rho_E )$.
These values are known to be bounded from above as \cite{Hayashi3,H07}
\begin{eqnarray}
\left\|\rho_{A,E}-\rho_A \otimes \rho_E \right\|_1
&\le &
2 \sqrt{2} \sqrt{P_{\rm ph}} \Label{Ha3}\\
\chi
&\le&  \eta_n(P_{\rm ph}), \Label{Ha4}
\end{eqnarray}
where $P_{\rm ph}$ is the phase error probability of the channel $\Lambda_t$.
That is, $P_{\rm ph}:=1-P^X(x=0^n)$.
The function $\eta_n$ is defined as 
\begin{equation}
\eta_n(x)
:=
\left\{
\begin{array}{ll}
-x \log x-(1-x)\log (1-x) + n x & \hbox{ if }x \le 1/2  \\
1 + n x & \hbox{ if }x > 1/2. 
\end{array}
\right.
\end{equation}

Now we turn to the security of the secret key.
The only difference here is that the key is effectively sent through the quantum channel that is error-corrected by the quantum CSS code corresponding to the classical CSS code $C_1$, $C_2$.
Hence by using essentially the same argument as above, the security can be evaluated by the phase error probability that remains after the quantum error correction.
When one sees it in the phase basis (i.e., the $x$ basis),
this probability is given by the decoding error probability of the classical CSS code $C_2^{\perp}/C_1^{\perp}$,
%(c.f., the latter half of the previous section).
which we denote by $P_{\rm ph}\left(C_2^{\perp}/C_1^{\perp}\right)$.
Then the security of the secret key can be evaluated as \cite{Hayashi3,H07}
\begin{eqnarray}
\left\|\rho_{A,E}-\rho_A \otimes \rho_E \right\|_1
&\le &
2 \sqrt{2} \sqrt{P_{\rm ph}\left(C_2^{\perp}/C_1^{\perp}\right)},  \Label{Ha3-1}\\
\chi
&\le & \eta_l\left(P_{\rm ph}(C_2^{\perp}/C_1^{\perp})\right). \Label{Ha4-1}
\end{eqnarray}
The same evaluation as (\ref{Ha3-1})
has been done by Renes\cite[Theorem 5.1]{Renes10}.
For the case of $C_1=\FF_2^n$, essentially the same relation was noted by Koashi \cite{Koashi} and Miyadera \cite{Miyadera}.

Then we apply Theorem \ref{thm4}
%the results of the previous section 
to evaluate $P_{\rm ph}\left(C_2^{\perp}/C_1^{\perp}\right)$.
%The main issue of this section is to evaluate the average performance of this phase error correction when the subcode $C_2 \subset C_1$ is randomly chosen from an $\varepsilon$-almost {\it dual} universal subcode family ${\cal C}$ of a fixed code $C_1$.
In our BB84 protocol, the subcode $C_2 \subset C_1$ is randomly chosen from an $\varepsilon$-almost {\it dual} universal subcode family ${\cal C}$ with minimum dimension $m-l$ of a fixed code $C_1$.
This corresponds to the case where the dual code $C_2^{\perp}$ is chosen from the $\varepsilon$-almost universal$_2$ extended code family of the fixed code $C_1^{\perp}$ with maximum dimension $n-m+l$.
%evaluate the average performance when the dual code $C_2^{\perp}$ is chosen with the equal probability from an $\varepsilon$-almost universal$_2$ extended code family ${\cal C}$ of the fixed code $C_1^{\perp}$.
Thus by applying inequality (\ref{Ha17}), we have
% on the decoding error probability of the classical CSS code.
\begin{align}
\rE_{C_2 \in {\cal C}} P_{\rm ph}\left(C_2^{\perp}/C_1^{\perp}\right)
\le \varepsilon 
\sum_{k=0}^{n} \hat{P}^X(k)
2^{-n\left[S-h\left(\min \{k/n,1/2\} \right)\right]_+}
%\frac{|C_2^{\perp}|}{2^{n}}
,\Label{Ha6}
\end{align}
where $S=(m-l)/n$ is the sacrificed bit rate, i.e. the ratio of bits reduced by privacy amplification.
Therefore, from (\ref{Ha3}), (\ref{Ha4}), and from the concavity of $x \mapsto \sqrt{x}$,
 $x\mapsto \eta_{l}$, we have
\begin{align}
& \rE_{C_2 \in {\cal C}}
\left\|\rho_{A,E}-\rho_A \otimes \rho_E \right\| \nonumber \\
\le &
2 \sqrt{2} \sqrt{
\varepsilon 
\sum_{k=0}^{n}
\hat{P}^X(k)
2^{-n\left[S-h\left(\min \{k/n,1/2\} \right)\right]_+}
%\frac{|C_2^{\perp}|}{2^{n}}
} ,
\Label{1-10-1}
\\
& \rE_{C_2 \in {\cal C}}\,
\chi \nonumber \\
\le &
\eta_l
\left(
\varepsilon 
\sum_{k=0}^{n} \hat{P}^X(k) 
2^{-n\left[S-h\left(\min \{k/n,1/2\} \right)\right]_+}
%2^{nh(\frac{k}{n})} \frac{|C_2^{\perp}|}{2^{n}}
\right).
\Label{1-10-2}
\end{align}
In practical QKD systems, the weight distribution $\hat{P}^X$ needs to be estimated from the bit error rate of sampled bits (see, e.g., \cite{Hayashi3,H07}).
If the phase error rate $p_{\rm ph}=k/n$ is estimated to be less than a certain value $\hat{p}_{\rm ph}$ with the exception of a negligiblly small probability,
and if $S > h(\hat{p}_{\rm ph})$,
then the argument $
\varepsilon 
2^{-n\left[S-h\left(\min \{k/n,1/2\} \right)\right]_+}$ 
converges to zero for $n\to\infty$.
Asymptotically, it is sufficient to sacrifice $n \left[h\left(\hat{p}_{\rm ph}\right)+\delta\right]$ bits by privacy amplification
%with the observed phase error probability $\hat{p}_{ph}$ and
with an arbitrary $\delta>0$.

From the above argument, we see that for the security of QKD, it is sufficient to choose the code $C_2$ from 
an $\varepsilon$-almost {\it dual} universal$_2$ subcode family of $C_1$,
while the existing results (e.g., \cite{Renner}) guarantee the security only when 
the code $C_2$ is randomly chosen 
from a universal$_2$ subcode family of $C_1$.
Since a universal$_2$ subcode family of $C_1$
is a $2$-almost dual universal$_2$ subcode family of $C_1$ (Theorem \ref{thm:excode-dual}),
our condition %for the code family of $C_2$
is strictly weaker than that by \cite{Renner}.

It should also be noted that by setting $C_1=\FF_2^n$, our argument also applies to  Koashi's proof technique \cite{Koashi}; that is, random matrices appearing in Koashi's protocol can be replaced by an almost dual universal$_2$ code family.

Further, the above discussion can be extended to 
an $\varepsilon$-almost {\it dual} universal$_2$ subcode pair family of 
$\{C_2 \subset C_1\}$.
Now, we choose 
$m-l$ dimensional subcode $C_2$ of $C_1$ such that
the dual code $C_2^{\perp}$ satisfies the condition of Theorem \ref{ht2-2}.
When the Pauli channel is permutation invariant,
this code satisfies (\ref{1-10-1}) and (\ref{1-10-2}) with $\varepsilon=n+1$.

\section{Quantum wire-tap channel}\Label{s6}
\subsection{Evaluation by phase error correction approach}\Label{s6-a}
We apply our results of the previous section on QKD for showing the security in the quantum wire-tap channel model.
In this model, the channel from Alice to Bob and the channel from Alice to Eve are both specified.
Particularly, in this section, we assume that the channel from Alice to Bob is given by the $n$-multiple use of the Pauli channel which is described by the joint distribution $P^{ZX}$ of bit error and phase error on a single qubit system.
We also assume that phase error and bit error occur independently,
and denote the phase error probability by $p_{\rm ph}$.
This corresponds to a limited case of the Pauli channel discussed in the previous section, i.e., $P^{X^nZ^n}(x,z)=\prod_{i=1}^nP^X(x_i)P^Z(z_i)$ with $P^X(1)=1-P^X(0)=p_{\rm ph}$.
As to the channel to Eve, we assume that Eve can access all part of the environment system corresponding to this channel.

Our goal is to show that Alice can send secret classical information via the quantum channel to Bob by the following coding protocol (c.f. the paragraph below (\ref{Ha20-2})).
% For given two codes $C_2 \subset C_1$,
% Alice encodes her information to $C_1/C_2$. 
% That is, in order to send the classical information $[x] \in C_1/C_2$,
% Alice randomly chooses one element $y \in C_2$ and sends $x+y$.
% When Bob can decode Alice's information concerning the code $C_1$ via this channel effectively,
% he can decode Alice's information concerning the code $C_1/C_2$ via this channel effectively.
First, Alice chooses a classical CSS code $C_1$, $C_2$.
A message to be sent is a coset $[x]\in C_1/C_2$, and when the sender wants to send $[x]$, she chooses an element randomly from the set $x+C_2$ with the equal probability and sends it.
On the receiver's side, Bob first applies the maximum likelihood decoder of $C_1$ on the received bit sequence and obtains an element $y\in C_1$. Then, he obtains a coset $[y] \in C_1/C_2$ as the final decoded message.
% Due to Shor-Preskill--type discussion \cite{SP00},
% the code $C_1/C_2$ on the bit basis corresponds 
% the code $C_2^{\perp}/C_1^{\perp}$ on the phase basis.
% So, in the following, we focus only on the relation between the choice of $C_2 (\subset C_1)$ 
% and Eve's information concerning the code $C_1/C_2$.

From Eve's point of view, this protocol is equivalent to the situation where Alice sends her classical information $[x]\in C_1/C_2$ by encoding it to a state $|[x]\rangle$ of the quantum CSS code (see, e.g., \cite{SP00}).
Hence we can evaluate the security of $[x]$ by the same argument as the previous section, i.e., by inequality (\ref{Ha3-1}) or by (\ref{Ha4-1}), depending on one's security criteria.
By noting that the channel between Alice and Bob is i.i.d., we can apply a simple bound given in Theorem \ref{thm4}.
Thus, if a fixed a code $C_1$, and an $\varepsilon$-almost dual universal$_2$ subcode family of ${\cal C}$ of $C_1$ are used, the average of $P_{\rm ph}\left(C_2^{\perp}/ C_1^{\perp} \right)$ satisfies
\begin{align}
\rE_{C_2 \in {\cal C}}
P_{\rm ph}\left(C_2^{\perp}/ C_1^{\perp} \right)
\le 
2^{-nE(1-S, p_{\rm ph})}
\max\{\varepsilon,1\}.
\Label{Ha7}
\end{align}
Here $t_{\min}=n(1-S)$ is the minimum dimension of $C_2$, and  
$t_{\max}=nS$ is the maximum dimension of $C_2^{\perp}$, which equals the sacrificed bit length.
As one can see from Corollary \ref{crl:ERp}, the exponential decreasing rate $E\left(1-S,p_{\rm ph}\right)$ on the right hand side of (\ref{Ha7}) is strictly positive for $S>h(p_{\rm ph})$.
% When $\varepsilon$ is bounded by a fixed constant
% and the sacrifice bit rate $R := \frac{t_{\min}}{n}$
% is greater than $h(p_{ph})$,
% the exponential decreasing rate $s R- (1+s)\psi(\frac{1}{1+s})$ is strictly positive 
% because $\frac{d\psi(a)}{da}|_{a=1}=h(p_{ph})$.
By using (\ref{Ha7}),
the averages of Eve's distingushability
$\|\rho_{AE}-\rho_{A} \otimes \rho_{E} \|_1$
and the Holevo information
$\chi=\Tr \rho_{AE} \left(\log \rho_{AE}-\log \rho_{A} \otimes \rho_{E} \right)$
can be evaluated as
\begin{align}
& \rE_{C_2\in {\cal C}} 
\|\rho_{AE}-\rho_{A} \otimes \rho_{E} \|_1 \nonumber \\
\le &
2^{-\frac12nE(1-S, p_{\rm ph})+\frac32}
\max\left\{\sqrt{\varepsilon},1\right\},
\Label{Ha21}\\
& \rE_{C_2\in {\cal C}}\, \chi \nonumber \\
\le & \eta_{l}\left(
2^{-nE(1-S, p_{\rm ph})}
\max\{\varepsilon,1\}\right).
\Label{Ha22}
\end{align}
with $l=\dim C_1-t_{\min}$ being the length of message.

\subsection{Deterministic universal hash function}\Label{s6-a2}
In fact, 
the above argument is valid even for 
a $\varepsilon$-almost dual universal$_2$ code pair family. 
%$\C_2\subset C_1$.
Since 
our setting is permutation invariant,
a deterministic code pair given in Proposition \ref{2-7-3} can be used.
That is,
given a code $C_1$, we can choose another $t$-dimensional subcode $C_2$ such that
$C_1^{\perp} \subset C_2^{\perp}$ and $\varepsilon(C_2^{\perp}/C_1^{\perp}) \le n+1$.
Then by combining (\ref{Ha7-2}), (\ref{Ha3-1}), and (\ref{Ha4-1}), we see that the security of $C_1$, $C_2$ can be evaluated as
\begin{align}
\|\rho_{AE}-\rho_{A} \otimes \rho_{E} \|_1
\le &
\sqrt{n+1}\ 
2^{-\frac12nE(1-S,\,p_{\rm ph})+\frac32},
\Label{Ha21-1}\\
\chi
\le & \eta_{l}\left(
(n+1)\,
2^{-nE(1-S,\,p_{\rm ph})}
\right)
\Label{Ha22-1}
\end{align}
with the message length $l=\dim C_1-t$.
Note that the construction of code $C_2$ is universal in that it does not depend on the value of $p_{\rm ph}$.
Hence, the linear map defined by $C_1\to C_1/C_2$ can be regarded as a type of deterministic universal hash function which is secure for 
independent and identical applications of
an arbitrarily given quantum Pauli channel.

\subsection{Comparison with $\delta$-biased approach}\Label{s6-b}
%Indeed, the above evaluation is valid even with the QKD scenario given in Section \ref{s5}
%if the quantum channel is the memoryless channel of the Pauli channel with the phase error probability $p_{\rm ph}$.

% where $l:= \dim C_1-t_{\max}= t_{\min}- \dim C_1^{\perp}$.
% So, the sacrifice bit rate $R$ is greater than $h(p)$, both values go to zero exponentially.

Now, we treat the same setting as the above by using
the $\delta$-biased approach.
When 
the subcode $C_2 \subset C_1$ 
is chosen from
an $\varepsilon$-{\it almost dual universal$_2$ subcode} family 
${\cal C}$
of a fixed code $C_1$,
we can evaluate the average performance after 
the combination of the error correction by $C_1$
and the privacy amplification by $C_2$
by using Lemma \ref{Lem6-3-q} (the $\delta$-biased approach).
%The method by Lemma \ref{Lem6-3-q} is based on $\delta$-biased while our method is based on the phase error correction.

When $\varepsilon\ge 1$,
attaching the smoothing to Lemma \ref{Lem6-3-q},
Hayashi \cite{cq-security} derived the following inequalities:
\begin{align}
& \rE_{C_2\in {\cal C}} 
\|\rho_{AE}-\rho_{A} \otimes \rho_{E} \|_1 \nonumber \\
\le &
(4+(n+1)^{1/2} \sqrt{\varepsilon})
2^{-\frac12nE(1-S,\,p_{\rm ph})} 
\Label{1-3-a}
\\
& \rE_{C_2\in {\cal C}} 
\chi \nonumber \\
\le &
\eta_{n}\left(
(4+(n+1)^{1/2} \sqrt{\varepsilon})
2^{-\frac12nE(1-S,\,p_{\rm ph})} \right) 
\Label{1-3-b}
\\
& \rE_{C_2\in {\cal C}} \chi \nonumber \\
\le &
2 \eta_{u(\varepsilon,n)}
( 2^{ 1-n  \max_{0\le s \le 1} \frac{s}{2-s} (S  -H_{1-s}(p_{\rm ph}))} ),
\Label{1-3-c}
\end{align}
where
$
H_{1-s}(p):= \frac{1}{s}\log_2 (p^{1-s}+(1-p)^{1-s})$
and 
$u(\varepsilon,n):=\frac{\varepsilon (n+1)}{4 \log 2} + n $.
When $\varepsilon$ increases at most polynomially,
(\ref{Ha21}) and (\ref{1-3-a}) give the same exponential evaluation:
\begin{align}
\liminf_{n \to \infty}\frac{-1}{n}\log \rE_{C_2} 
\|\rho_{AE}-\rho_{A} \otimes \rho_{E} \|_1
&\ge
\frac12 E(1-S,\,p_{\rm ph}) .
\end{align}
However, for $\varepsilon\ge 1$,
\begin{align}
\frac{\hbox{RHS of }(\ref{Ha21})}{\hbox{RHS of }(\ref{1-3-a})}
=
\frac{2^{3/2} \sqrt{\varepsilon}}{(4+(n+1)^{1/2} \sqrt{\varepsilon})}
\to 0.
\end{align}
Hence, we can conclude that the evaluation (\ref{Ha21}) 
by the phase error correction approach
gives a better evaluation for 
$\|\rho_{AE}-\rho_{A} \otimes \rho_{E} \|_1 $.

In this case, 
(\ref{1-3-b}) yields the following exponential evaluation
for $\chi$:
\begin{align}
\liminf_{n \to \infty}\frac{-1}{n}\log \rE_{C_2\in {\cal C}} 
\chi
\ge 
\frac12 E(1-S,\,p_{\rm ph}) ,
\end{align}
which is
better than that of (\ref{1-3-c}), as is shown in Hayashi \cite{cq-security}.
However, the evaluation (\ref{Ha22}) by the phase error correction approach
gives the following:
\begin{align}
\liminf_{n \to \infty}\frac{-1}{n}\log \rE_{C_2\in {\cal C}} 
\chi
\ge 
E(1-S,\,p_{\rm ph}) ,
\end{align}
which is twice of the above.
Hence, 
in the case of QKD, 
we can conclude that 
the phase error correction approach
is better than
the $\delta$-biased approach
based on Lemma \ref{Lem6-3-q}.

\section{Relation with existing results}\Label{s8}

\subsection{Comparison with existing results}\Label{s8-a}
In order to compare our results of this section with existing ones, we here review the history of the studies of the information theoretic security.
%For a sacrifice bit rate $R$ greater than
%the mutual information $I(A:E)$ between Alice and Eve,

Wyner \cite{Wyner}, and Csisz\'{a}r and K\"{o}rner \cite{CK79} 
showed the weak security 
with the wire-tap channel model
in terms of Maurer and Wolf \cite{MW}.
Csisz\'{a}r \cite{Csiszar} showed the strong security
with the same model
in terms of Maurer and Wolf \cite{MW}.
Hayashi \cite{Hayashi1} gave the concrete exponential decreasing rate for the strong security
with the same model.
These studies use completely random coding as privacy amplification process. 
That is, no linear functions are used in this process.
Bennett et al. \cite{BBCM} 
and 
H\r{a}stad et al. \cite{ILL}
proposed to use universal$_2$ hash functions for privacy amplification.
Maurer and Wolf \cite{MW} applied this idea to the secret key agreement, which is different setting form wire-tap channel.
They showed the strong security with universal$_2$ hash functions for privacy amplification.
Based on these ideas, Hayashi \cite{Hayashi2} showed the strong security 
with universal$_2$ hash functions
when the sacrifice bit rate is greater than the mutual information $I(A:E)$.
Muramatsu and Miyake \cite{MM2} considered
a more general condition \cite{MM1}
than the $\varepsilon$-almost universal$_2$ functions of the code for privacy amplification.
Under this condition, they showed the weak security 
%i.e., $\rE_{C_2^{\perp}\in {\cal C}} I(A,E)/n \to 0$.
However, 
Watanabe et al. \cite{WMU} pointed out that their method cannot derive the strong security
based on Hayashi's idea \cite{Hayashi4} in the case of secret key agreement from correlated source.
Further,
the impossibility of the strong security under the condition of 
$\varepsilon$-almost universal$_2$ 
 will be shown in Theorem \ref{thm:counterexample} by giving a counterexample.
Overall, our concept ``$\varepsilon$-almost universal$_2$"
is a larger class of hash function families 
than any known classes of linear hash function families guaranteeing the strong security.

\subsection{$\varepsilon$-almost {\it dual} universality$_2$ vs. $\varepsilon$-almost universality$_2$}
\Label{s8-b}
Finally, as mentioned earlier,
we present an example of the classical wire-tap channel model that can vividly contrast the properties of the $\varepsilon$-almost {\it dual} universality$_2$ and the $\varepsilon$-almost universality$_2$.
Tomamichel et al. 
showed that 
when $\varepsilon$ converges to $1$,
any sequence of 
$\varepsilon$-almost universal$_2$ subcode families (of $C_1=\FF_2^n$) 
guarantees the strong security\cite[Lemma 1]{TSSR11}\footnote{Their $\delta$
corresponds to $\varepsilon 2^{m}$ when the bit length of final keys
is $m$.}.
However, 
one sees that, if $\varepsilon\ge2$, an $\varepsilon$-almost universal$_2$ subcode family (of $C_1=\FF_2^n$) cannot necessarily guarantee the strong security.
%is not guaranteed even the code $C_2$ is chosen from an $\varepsilon$-almost universal$_2$ subcode family of $C_1=\FF_2^n$ with $\varepsilon\ge2$, i.e.,
%an $\varepsilon$-almost universal$_2$ code family of $\FF_2^n$ with $\varepsilon\ge2$.
%In this case, an $\varepsilon$-almost dual universal$_2$ subcode family of $C_1^{\perp}$ is simply an $\varepsilon$-almost dual universal$_2$ code family.
In other words, the choice of the code $C_2$ from an $\varepsilon$-almost universal$_2$ subcode family of $C_1$ is not sufficient for the strong security.
Note that we have shown in this section that the $\varepsilon$-almost {\it dual} universality$_2$ is indeed sufficient for this purpose.
Hence, at least in the setting of this section, the $\varepsilon$-almost {\it dual} universality$_2$ is the more relevant criterion for security.

\begin{Thm}\label{thm:counterexample}
Assume that the channel from Alice to Bob is noiseless,
and the channel to Eve is binary symmetric with error probability $p$.
There exists an example of a 2-almost universal$_2$ code family ${\cal C}$ for which the hash functions (i.e., $\FF_2^n\to \FF_2^n/C_2$ with $C_2\in{\cal C}$) cannot guarantee the strong security.
\end{Thm}

\begin{proof}
Choose an arbitrary universal$_2$ code family ${\cal C}'=\{C_2' \subset \FF_2^{n-1}\}$.
Then define another code family ${\cal C}$ in $\FF_2^n$, consisting of $C_2:=\left\{\,x||0\,|\,x\in C_2'\,\right\}$ with $C_2'\in{\cal C}$.
Here, $a\|b$ denotes the concatenation of $a$ and $b$.
Hence for any $C_2\in{\cal C}$, there exists $C_2'\in {\cal C}'$, such that 
$C_2$ consists of $x\in C_2'$ concatenated with a zero.
Note that the code family ${\cal C}$ is obviously $2$-almost universal$_2$,
but its dual code family ${\cal C}^\perp$ cannot be $\varepsilon$-almost universal$_2$ for any $\varepsilon<1$, because $x=0\dots01\in C$ for all $C\in{\cal C}^\perp$.

When Alice transmits a coset $[x]\in\FF_2^n/C_2$ as her secret message, 
she chooses $x\in[x]$ randomly and sends it to Bob.
Due to our construction of ${\cal C}$, the $n$-th bit of $x$ is preserved in $[x]$ as it is without being canceled by privacy amplification.
Since Eve receives this $n$-th bit with the error probability $p$,
Eve's mutual information regarding $[x]$ is greater than $1-h(p)$.
Therefore, the strong security does not hold with these hash functions.
\end{proof}

\subsection{Deterministic universal hash function}\Label{s8-c}
%\subsection{Deterministic universal hash function}\Label{s6-a2}
When there exist errors, one needs error correction as well as
 hash functions.
Here we denote the code for error correction by $C_1$
and the code for the hash function by $C_2$.
Then, the relation $C_2\subset C_1$ holds.
Now, we consider what kind of code pairs $C_2\subset C_1$
yields the strong security.

First note that the phase error correction approach
has an additional advantage over 
the $\delta$-biased approach;
that is, the phase error correction approach allows us to use
an $\varepsilon$-{\it almost dual universal$_2$ code pair} family 
$C_2\subset C_1$.

Note also that the situation is quite different for the $\delta$-biased approach,
because it requires hash functions to be applied after error correction.
That is, one needs to perform an $\varepsilon$-almost dual universal$_2$ code family 
to a fixed code space.
Hence, the $\delta$-biased approach can guarantee the strong security only with 
an $\varepsilon$-{\it almost dual universal$_2$ subcode} family of a fixed code $C_1$.
This relation among classes of code pair families are summarized in Fig. \ref{fig:code-pair}.

\begin{figure*}[!t]

\begin{center}
\scalebox{1.0}{\includegraphics[scale=0.5]{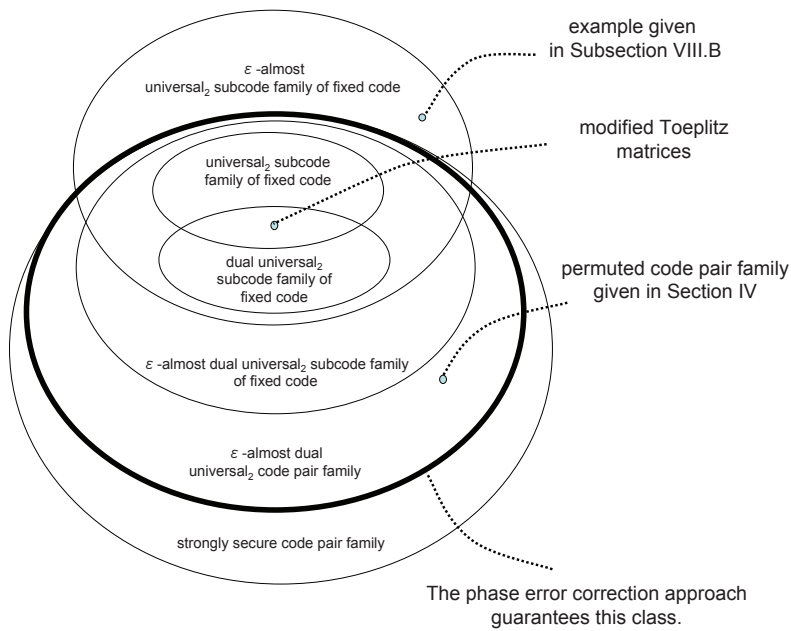}}
\end{center}
\caption{Relation among class of code pairs.}
\Label{fig:code-pair}

\end{figure*}

In order to illustrate this advantage of the phase error correction approach with an example, let us take an arbitrary code $C_2$,
and choose a subcode $C_1$ of $C_2$ based on Proposition \ref{2-7-3}.
Then,
the permuted code pair family ${\cal C}_{C_2\subset C_1}$
is an $(n+1)$-{\it almost dual universal$_2$ code pair} family,
but %does not has a form of
is not an $(n+1)$-{\it almost dual universal$_2$ subcode} family
of a fixed code $C_1$.
Hence, 
as is discussed in Subsection \ref{s6-a2},
for a given error correction code $C_1$, 
the phase error correction approach guarantees the existence of
a deterministic hash function that universally works for 
an independent and identical setting.
In particular, 
if the error correcting code $C_1$ 
universally works for additive errors given by an
independent and identical distribution,
the code pair $C_2\subset C_1$
universally works for error correction as well as privacy amplification. 

However, in the $\delta$-biased approach,
it is impossible to construct such a deterministic hash function
because this approach cannot treat the security for 
an $(n+1)$-{\it almost dual universal$_2$ code pair} family.

%\subsection{Relation with Hamada \cite{Ham}}\Label{s8-d}

Finally, we explain the relation of our results to a universal quantum CSS code found by Hamada \cite{Ham} for sending quantum states.
In his paper, he focused on an family of classical self-dual codes.
Then combining qubits based on the bit basis and qubits based on the phase basis,
he succeeded in constructing a universal quantum CSS code from a set of universal classical self-dual codes by choosing $C_1^\perp=C_2$.
His code can be applied to QKD, where Alice can send information by using both of the bit basis and the phase basis. 
On the other hand, it cannot be applied to our quantum wire-tap channel model in a straightforward manner, where only the bit basis is used for sending the classical message.
This is because our method employs two codes $C_1$ and $C_2$ chosen separately.
%A possible way around would be to find a universal code among non-self-dual codes.
%Due to the same problem, 
Our method for constructing a deterministic universal hash function would not work either, if we were to restrict our codes to self-dual codes.
%we cannot construct  for classical wire-tap channel and randomness extraction.
Recall that the key point of our method is the concept of a ``permuted code pair family."
%Using this idea, we can construct a universal code for these models.

\section{Conclusion}
In this paper, we have first introduced the concept of
``$\varepsilon$-almost dual universal$_2$ hash function family".
Then, 
we have shown that
the class of $\varepsilon$-almost dual universal$_2$ hash function families
includes the class of universal$_2$ hash function families.

Employing the relation between quantum error correction and the security,
we have shown that application of $\varepsilon$-almost dual universal$_2$ hash function family
yields the strong security.
We have also mentioned that the results concerning the $\delta$-biased family
\cite{DS05,FS08} imply this fact,
while their original result does not refer the privacy amplification.

We have compared these two approaches, i.e.,
the phase error correction approach and the
$\delta$-biased approach in the following two points.
As the first point, we have shown that
the phase error correction approach 
yields a better security bound 
in terms of the trace distance and the Holevo information, 
than the $\delta$-biased approach.
As the second point, we have shown that
the phase error correction approach 
guarantees the strong security 
with a larger class of protocols
than the $\delta$-biased approach
when we apply error correction as well as privacy amplification.

In particular, as a byproduct, 
we have shown the existence of a universal code for privacy amplification with error correction.
Due to the above difference, the phase error correction approach can guarantee the existence of such a code,
while the $\delta$-biased approach cannot.

\section*{Acknowledgments}
The authors are grateful to the
referee of the previous version for informing the literatures \cite{DS05,FS08}
and giving a crucial comment for Lemma \ref{Lem6-3-q}.
MH is grateful to Dr. Markus Grassl for a helpful discussion.
TT and MH are partially supported by the National Institute of Information and Communication Technolgy (NICT), Japan.
MH is partially supported by a MEXT Grant-in-Aid for Young Scientists (A) No. 20686026 and Grant-in-Aid for Scientific Research (A) No. 23246071.
The Centre for Quantum Technologies is funded by the
Singapore Ministry of Education and the National Research Foundation
as part of the Research Centres of Excellence programme.

\appendices

\section{Proofs of Theorems \ref{thm3} and \ref{thm4}}\label{s7}
First, we show Theorem \ref{thm3}.
Due to the linearity, it is sufficient to evaluate the probability that 
the received signal is erroneously decoded to $C \setminus \{0\}$ 
when $0 \in C$ is sent.
Let $P^n_X(x)$ be the $n$-independent and identical extension of the distribution $(1- p,p)$.
Since the phase error $x$ occurs on $n$-bits sequence with the probability $P^n_X(x)$,
applying Gallager's evaluation\cite{Gal} to this error probability,
for $0\le s \le 1$ and $0\le a= \frac{1}{1+s}$, 
we obtain 
\begin{align*} 
& P_e(C) \le
\sum_{y\in \FF_2^n}
P^n_X(y)
\left( 
\sum_{x \in C \setminus \{0\} }
\left(\frac{P^n_X(y+x)}{P^n_X(y)}\right)^a
\right)^s \\
= &
\sum_{y\in \FF_2^n}
(P^n_X(y))^{\frac{1}{1+s}}
\left( \sum_{x \in C \setminus \{0\} }
(P^n_X(y+x))^{\frac{1}{1+s}}
\right)^s.
\end{align*}
Thus, the error probability 
$P(C)$ is bounded from above by this value. 
Any $\varepsilon$-almost universal$_2$ code family satisfies the inequality
$\rE_{C\in{\cal C}} \sum_{x \in C \setminus \{0\}}
P^n_X(y+x)^{\frac{1}{1+s}}
\le  
\varepsilon 2^{t_{\max}-n}
\sum_{x \in \FF_2^n }
P^n_X(y+x)^{\frac{1}{1+s}}$.
Taking the average concerning the family for $C$, 
we obtain the upper bound
\begin{align}
& \rE_{C\in{\cal C}} P_e(C) \nonumber \\
\le &
\rE_{C\in{\cal C}} \sum_{y\in \FF_2^n}
P^n_X(y)^{\frac{1}{1+s}}
\left( \sum_{x \in C \setminus \{0\}}
P^n_X(y+x)^{\frac{1}{1+s}}
\right)^s \nonumber \\
\le &
\sum_{y\in \FF_2^n}
P^n_X(y)^{\frac{1}{1+s}}
\left( \rE_{C\in{\cal C}} \sum_{x \in C \setminus \{0\}}
P^n_X(y+x)^{\frac{1}{1+s}}
\right)^s \nonumber \\
\le &
\sum_{y\in \FF_2^n}
P^n_X(y)^{\frac{1}{1+s}}
\left( 
\varepsilon 2^{t_{\max}-n}
\sum_{x \in \FF_2^n }
P^n_X(y+x)^{\frac{1}{1+s}}
\right)^s ,\Label{Ha12}
\end{align}
where 
the concavity of $x \mapsto x^s$ is used.
Since the quantity
$\left( 
\varepsilon 2^{t_{\max}-n}
\sum_{x \in \FF_2^n }
P^n_X(y+x)^{\frac{1}{1+s}}
\right)^s$
does not depend on $y$, it can be replaced with
$\left( 
\varepsilon 2^{t_{\max}-n}
\sum_{x \in \FF_2^n }
P^n_X(x)^{\frac{1}{1+s}}
\right)^s
=
\varepsilon^s 2^{st_{\max}-sn}
\left( 
\sum_{x \in \FF_2^n }
P^n_X(x)^{\frac{1}{1+s}}
\right)^s$.
Hence, the right hand side of (\ref{Ha12}) becomes
\begin{align}
& \sum_{y\in \FF_2^n}
P^n_X(y)^{\frac{1}{1+s}}
\varepsilon^s 2^{st_{\max}-sn}
\left( 
\sum_{x \in \FF_2^n }
P^n_X(x)^{\frac{1}{1+s}}
\right)^s \nonumber \\
=&
\varepsilon^s 2^{st_{\max}-sn}
\left(
\sum_{x \in \FF_2^n }
P^n_X(x)^{\frac{1}{1+s}}
\right)^{1+s} \nonumber \\
=&
\varepsilon^s 2^{st_{\max}-sn}
2^{n\left[s-E_0(s,p)\right]}.\Label{Ha16}
\end{align}From this, we obtain Theorem \ref{thm3}.

Next, we show Theorem \ref{thm4}.
Due to the linearity, it is sufficient to evaluate the probability that 
the received signal is erroneously decoded to $C_1 \setminus C_2$ 
when Alice sends $0 \in C_2$.
The difference from the above case is the derivation of (\ref{Ha12}).
This part of derivation can be replaced as follows.
\begin{eqnarray*} 
&& \rE_{C_1\in{\cal C}} 
P_e(C_1 / C_2 ) \nonumber \\
&\le&
\rE_{C_1\in{\cal C}} \sum_{y\in \FF_2^n}
P^n_X(y)^{\frac{1}{1+s}}
\left( \sum_{x \in C_1 \setminus C_2}
P^n_X(y+x)^{\frac{1}{1+s}}
\right)^s \\
&\le&
\sum_{y\in \FF_2^n}
P^n_X(y)^{\frac{1}{1+s}}
\left( \rE_{C_1\in{\cal C}} \sum_{x \in C_1 \setminus C_2}
P^n_X(y+x)^{\frac{1}{1+s}}
\right)^s \\
&\le&
\sum_{y\in \FF_2^n}
P^n_X(y)^{\frac{1}{1+s}}
\left( 
\varepsilon 2^{t_{\max}-n}
\sum_{x \in \FF_2^n }
P^n_X(y+x)^{\frac{1}{1+s}}
\right)^s .
\end{eqnarray*}
Combining this and (\ref{Ha16}), we obtain (\ref{Ha7-1}) and (\ref{Ha7-1a}).
This discussion can be extended to the case of 
$\varepsilon$-almost universal$_2$ extended code pair family.
Thus, we obtain
Theorem \ref{thm4}.

\section{Proof of Equation (\ref{11-23-1})}
\label{app:Proof_min_DQP}
In order to prove this equation, it is convenient to introduce another binary distribution $P_\theta=(p_\theta,1-p_\theta)$ that is derived from $P=(p,1-p)$, where $p_\theta$ is defined by
\[
p_\theta
:= \frac{p^\theta}{p^\theta+(1-p)^\theta}
\]
with the convention that $p^0=0$ if $p=0$.
The distribution $P_\theta$, parameterized by a real number $\theta\ge0$, is often called the exponential family of $P$.
We also define a function $\psi(\theta)$ by
\[
\psi(\theta):=\log \left[p^\theta+(1-p)^\theta\right].
\]
Then the following relations are useful for simplifying calculations of divergence $d(p\|q)$ and entropy $h(p)$.
For $\theta \ge0$, we have
\[
\psi'(\theta)=-d(p_\theta\|p)-h(p_\theta),\ \ 
\psi''(\theta)
%=
%\frac{\sum_{x}p(x)^\theta (\log p(x))^2}
%{\sum_{x'} p(x')^\theta}-
%\psi'(\theta)^2
\ge 0,
\]
\begin{eqnarray}
h(p_\theta)
&=&
-\theta \psi'(\theta)+\psi(\theta), \nonumber \\
\frac{d h(p_\theta)}{d\theta}
&=& - \theta \psi''(\theta) \le 0, \nonumber
\end{eqnarray}
\begin{eqnarray}
d(p_\theta\|p)
&=&
-\psi(\theta) 
-(1-\theta) \psi'(\theta),\nonumber\\
d(p_\theta\|p)+h(p_\theta)
&=&
- \psi'(\theta).\nonumber
\end{eqnarray}
We shall make frequent use of these formulas in what follows.
Note that $E_0(s,p)$ can be rewritten as
\[
E_0(s,p)=s-(1+s)\psi\left(\frac1{1+s}\right).
\]

First, we prove Equation (\ref{11-23-1}) for the limited case where the minimum is evaluated over $q=p_\theta$ with $0\le\theta\le1$.
\begin{Lmm}
\label{Lmm:E_RP_P_theta}
If $R<1-h(p)$,
\begin{equation}
\min_{0\le \theta\le 1} d(p_\theta\|p)+[1 - h(p_\theta)-R]_+=E(R,p).
\label{eq:Lmm_E_RP_P_theta}
\end{equation}
\end{Lmm}
\begin{proof}
$E_R(s,p)=-sR+E_0(s,p)$ is convex with respect to $s$, since
$E_R''(s,p)=
(1+s)^{-3}\psi''\left(1/(1+s)\right) \ge 0$.
We define the critical rate $R_c$ by
\[
R_c:=1-h\left(p_{1/2}\right),
\]
such that, if $R\le R_c$ (resp., $R\ge R_c$), then $\left.\frac{\partial E_R}{\partial s}\right|_{s=1}\ge 0$ (resp., $\left.\frac{\partial E_R}{\partial s}\right|_{s=1}\le 0$).

Then, if $R\le R_c$, the maximum of $E_R$ is attained at $s=1$:
\begin{eqnarray*}
E(R,p)
&=&E_{R}(1,p)
=-R+1 -2 \psi(1/2)\\
&=&
d\left(\left.p_{1/2}\right\|p\right)+1-h(p_{1/2})-R\\
&=& \min_{0\le\theta\le1} d\left(\left.p_{\theta}\right\|p\right)+1-h(p_\theta)-R.
\end{eqnarray*}
The last line follows by noting that $d(p_\theta\|p)+1-h(p_\theta)-R$ attains its minimum at $\theta=1/2$, since  $\frac{\partial}{\partial \theta}\left[d(p_\theta\|p)-h(p_\theta)\right]=(\theta-1/2)\psi''(\theta)$ with $\psi''(\theta)\ge0$.
Also by noting that $1-h(p_{1/2})-R\ge 0$ for $R\le R_c$, we see that (\ref {eq:Lmm_E_RP_P_theta}) is satisfied for $R\le R_c$.

On the other hand, if $R>R_c$, we have $\left.\frac{\partial E_R}{\partial s}\right|_{s=1}\le 0$, and also $\left.\frac{\partial E_R}{\partial s}\right|_{s=0}>0$ from $R<1-h(p)$.
Thus the maximum is attained at $s_R\in(0,1]$ satisfying $\left.\frac{\partial E_R}{\partial s}\right|_{s=s_R}= 0$, i.e.,
\begin{equation}
\psi\left(\frac{1}{1+s_R}\right)
-\frac{1}{1+s_R}\psi'\left(\frac{1}{1+s_R}\right)
= 1 -R.
\label{eq:cond_sR}
\end{equation}
Hence
\begin{align}
& E(R,p) =E_{R}(s_R,p) \nonumber \\
= & -\psi\left(\frac{1}{1+s_R}\right) -\frac{s_R}{1+s_R} \psi'\left(\frac{1}{1+s_R}\right)\nonumber \\
=&d\left(\left.p_{(1+{s_R})^{-1}}\right\|p\right).
\Label{28-6}
\end{align}
Note that the condition (\ref{eq:cond_sR}) can also be written as  $1-h\left(p_{(1+s_R)^{-1}}\right)-R=0$.
Then by noting that $d(p_\theta\|p)-h(p_\theta)$ is monotonically increasing for $1/2\le \theta\le1$, whereas $d(p_\theta\|p)$ decreasing, we see that the minimum of 
(\ref{eq:Lmm_E_RP_P_theta}) is attained for $\theta=(1+s_R)^{-1}$.
Hence (\ref{eq:Lmm_E_RP_P_theta}) holds for $R>R_c$ as well.
\end{proof}

{\it Proof of Equation (\ref{11-23-1}):}
Let
\begin{eqnarray*}
M_1&:=&\min_{0\le q\le1}d(q\|p)+[1-h(q)-R]_+,\\
M_2&:=&\min_{0\le \theta\le 1} d(p_\theta\|p)+[1 - h(p_\theta)-R]_+.
\end{eqnarray*}
Then from Lemma \ref{Lmm:E_RP_P_theta}, it suffices to show $M_1=M_2$.
Since $M_1\le M_2$ holds trivially, it remains to show $M_1\ge M_2$.

Denote the value of $q$ attaining the minimum of $M_1$ by $\tilde{q}$.
Then we have
\begin{equation}
d(\tilde{q}\|p)\le d(p_0\|p)
\label{eq:D_tildeQ_inequality}
\end{equation}
since otherwise,
\begin{eqnarray}
M_1&>&d(p_0\|p)+[1-h(\tilde{q})-R]_+\nonumber\\
&\ge&d(p_0\|p)+[1-h(p_0)-R]_+\ge M_2,
\label{eq:M1_M2_contradiction}
\end{eqnarray}
which contradicts $M_1\le M_2$.
The second line of (\ref{eq:M1_M2_contradiction}) follows by noting that $h(\tilde{q})\le h(p_0)$ with $p_0$ being the uniform distribution.
Note that this is true even when $p=0$ (resp. $p=1$) because then $\tilde{q}=0$ (resp. $\tilde{q}=1$) due to the condition $d(\tilde{q}\|p)<\infty$. 

By a straightforward calculation, one can show that, given an arbitrary combination of $p,q,\theta$ satisfying $d(q\|p)=d(p_\theta\|p)$,
\begin{equation}
h(p_\theta)-h(q)
=\frac{d(q\|p_\theta)}{1-\theta}\
\label{eq:HP_theta-HQ}
\end{equation}
holds.From (\ref{eq:D_tildeQ_inequality}), $d(\tilde{q}\|p)=d(p_{\tilde{\theta}}\|p)$ holds for some $\tilde{\theta}\in[0,1]$.
Then by using (\ref{eq:HP_theta-HQ}), we see that $h(p_{\tilde{\theta}})\ge h(\tilde{q})$, and thus $M_1\ge M_2$.
\endproof

\end{document}